\documentclass[journal]{IEEEtran}

\usepackage[pdftex]{graphicx}
\usepackage[cmex10]{amsmath}
\usepackage{algorithmicx}
\usepackage{array}

\usepackage{amssymb,mathrsfs}
\usepackage{times}
\usepackage{color}
\usepackage{subfig}
\usepackage{cite}
\usepackage{multirow,tabularx}
\usepackage[ruled,lined]{algorithm2e}
\usepackage[compatible]{algpseudocode} 

\usepackage{wasysym}
\usepackage{stackrel}
\usepackage{tikz}
\usetikzlibrary{matrix,arrows,decorations.pathreplacing,automata,positioning,decorations.pathmorphing}
\usepackage{float}
\usetikzlibrary{shapes.multipart}
\newcommand{\myunit}{1 cm}
\tikzset{
    node style sp/.style={draw,circle,minimum size=\myunit},
    node style ge/.style={circle,minimum size=\myunit},
    arrow style mul/.style={draw,sloped,midway,fill=white},
    arrow style plus/.style={midway,sloped,fill=white},
}

\DeclareMathOperator*{\argmin}{arg\,min}

\newtheorem{theorem}{{\bf Theorem}}

\newtheorem{corollary}{{\bf Corollary}}

\newtheorem{example}{{\bf Example}}
\newtheorem{definition}{{\bf Definition}}
\newtheorem{remark}{{\bf Remark}}

\newcommand{\net}{\mathcal{N}}
\newcommand{\netnodes}{V}
\newcommand{\netedges}{E}

\newcommand{\compgraph}{\mathcal{G}}
\newcommand{\compnodes}{\Omega}
\newcommand{\compedges}{\Gamma}
\newcommand{\prenodes}[1]{\Phi_{\uparrow}(#1)}
\newcommand{\sucnodes}[1]{\Phi_{\downarrow}(#1)}

\newcommand{\distance}{\mathtt{d}} 
\newcommand{\edgewt}{\mathtt{W}} 
\newcommand{\processingwt}{\mathtt{P}} 

\newcommand{\mincost}{\textsf{MinCost}}
\newcommand{\mindelay}{\textsf{MinDelay}}

\newcommand{\embedding}{\mathcal{E}}
\newcommand{\setofembeddings}{\mathbb{E}}
\begin{document}

\title{Optimal Embedding of Functions for In-Network Computation:
  Complexity Analysis and Algorithms}

\author{Pooja Vyavahare,~\IEEEmembership{} Nutan
  Limaye,~\IEEEmembership{} and~D.~Manjunath,~\IEEEmembership{Senior
    Member,~IEEE}
  \thanks{Pooja Vyavahare and D.~Manjunath are with the Deptt. of
    Elecl. Engg., and Nutan Limaye is with the Deptt.~ of Comp.
    Sc.~\& Engg. of IIT Bombay.}  \thanks{Email addresses are
    \{vpooja,dmanju\}@ee.iitb.ac.in and {nutan}@cse.iitb.ac.in }
  \thanks{This work was supported in part by DST projects SR/SS/EECE/0139/2011, ITRA/15(64)/Mobile/USEAADWN/01, and was
    carried out in the Bharti Centre for Communication at IIT Bombay.}
}

\maketitle

\begin{abstract}
  We consider optimal distributed computation of a given function of
  distributed data. The input (data) nodes and the sink node that
  receives the function form a connected network that is described by
  an undirected weighted network graph. The algorithm to compute the
  given function is described by a weighted directed acyclic graph and
  is called the computation graph. An embedding defines the
  computation communication sequence that obtains the function at the
  sink. Two kinds of optimal embeddings are sought, the embedding
  that---(1)~minimizes delay in obtaining function at sink, and
  (2)~minimizes cost of one instance of computation of function. This
  abstraction is motivated by three applications---in-network
  computation over sensor networks, operator placement in distributed
  databases, and module placement in distributed computing.

  We first show that obtaining minimum-delay and minimum-cost
  embeddings are both NP-complete problems and that cost minimization
  is actually MAX SNP-hard. Next, we consider specific forms of the
  computation graph for which polynomial time solutions are
  possible. When the computation graph is a tree, a polynomial time
  algorithm to obtain the minimum delay embedding is described. Next,
  for the case when the function is described by a layered graph we
  describe an algorithm that obtains the minimum cost embedding in
  polynomial time. This algorithm can also be used to obtain an
  approximation for delay minimization. We then consider bounded
  treewidth computation graphs and give an algorithm to obtain the
  minimum cost embedding in polynomial time.
\end{abstract}

\IEEEpeerreviewmaketitle

\section{Introduction}
\label{sec:intro} 

\subsection{Background}

Efficient computing of functions of distributed data is of interest in
many applications, most recently in sensor networks \cite{Giridhar05}
and in programming models for distributed processing of large data
sets, e.g., MapReduce \cite{Dean04} and Dryad \cite{Isard07}.

Consider a typical sensor network scenario where sensor nodes measure
and store environment variables and form a distributed database. These
nodes also have computing and communication capabilities and can form
a connected network. A sink node, also called terminal or receiver, is
interested in one or more functions of this distributed data rather
than in the raw data itself. Conventional examples for such functions
are maximum, minimum, mean, parity, and histogram
\cite{Giridhar05}. More sophisticated functions are also easily
motivated, e.g., spatial and temporal correlations, spectral
characteristics of the data (that may be obtained by performing FFT on
the data), and filtering operations on the raw data in sensor
networks. A naive approach to obtain the required function(s) of the
distributed data at the sink would be to collect the raw data at the
sink and have it perform the computation. Alternatively, since the
nodes have computation capability, it could possibly be more efficient
to push the computation into the network, i.e., use a distributed
computation scheme over the communication network.  Our interest is in
the latter approach---efficient `in-network computation' of the
required function.

Several measures of efficiency of computation may be defined. Total
energy expended in obtaining one sample of the function is a possible
measure. Delay from the time at which the data is available at the
sources to the time at which the function value is available at the
sink is a second possible measure. If the data at each of the sources
were to be a stream, then the rate at which the function values are
available at the sink is a third possible measure. In this paper we
consider only the first two measures above---cost of computation and
delay in computation. Thus our focus is on algorithms that find an
computation and communication (routing) sequence to compute a target
function on a given network that minimize the delay or the cost.  The
target functions are assumed to belong to the class of functions that
are computed using a scheme that can be represented as directed
acyclic graphs (DAGs). The following example illustrates our intent.

Consider a network of $N$ nodes, $K$ of which collect measurement data
from their environment. Let $x_k(t)$ be data sample available at node
$k$ ($k=1, \ldots, K$) at time $t$ and let $r(t) = \frac{1}{K-1}
\sum_{i=1}^{K-1} x_{i}(t) x_{i+1}(t)$ be the function of
interest. $r(t)$ can be computed using the schema of
Fig.~\ref{fig:r-of-t}a. Each edge in this directed graph represents an
intermediate value in the computation of $r(t)$ and each node
corresponds to an operation that is to be performed on its inputs. The
communication network over which $r(t)$ is to be computed is shown in
Fig.~\ref{fig:r-of-t}b as a weighted undirected graph. In the example
all the edges have unit weight. Two possible computation and
communication schemes are shown in
Figs.~\ref{fig:r-of-t}c,d. We see that the scheme in
Fig.~\ref{fig:r-of-t}c has a lower cost than that in
Fig.~\ref{fig:r-of-t}d.

Several flavors of in-network function computation exist in the
literature. A randomly deployed multihop wireless network of noise
free links is considered in
\cite{Giridhar05,Giridhar06,Khude05,Kamath14}. They determine
asymptotic achievable rates at which different symmetric functions
like minimum, maximum and type vector may be computed. An identical
objective is addressed in
\cite{Gallager88,Kushilevitz98,Feige00,Newman04,Goyal05} for single
hop networks with noisy links and in \cite{Ying07,Dutta08} for
multihop wireless networks with noisy links. When the nodes are in a
$\sqrt{n} \times \sqrt{n}$ grid and communicate over wireline or
wireless links, \cite{Karamchandani11} obtains the time and number of
transmissions required to compute a function. Randomized gossip
algorithms \cite{Boyd05,Shah09}, where a random sequence of node-pairs
exchange data and perform a specific computation is used in
\cite{Mosk-Aoyama06,Ayaso08,Bodas11} for function computation. The
interest is in time for all nodes to converge to the specified
function. Another stream of work considers computing of specific
functions using network coding and designs the communication networks
to maximize the computation rate
\cite{Ramamoorthy10,Appuswamy11,Rai12}. None of the above consider
minimum cost function computation. Further, to the best of our
knowledge, only \cite{Shah13} considers computation of arbitrary
functions in finite size networks; the interest there is on maximizing
the rate of computation.

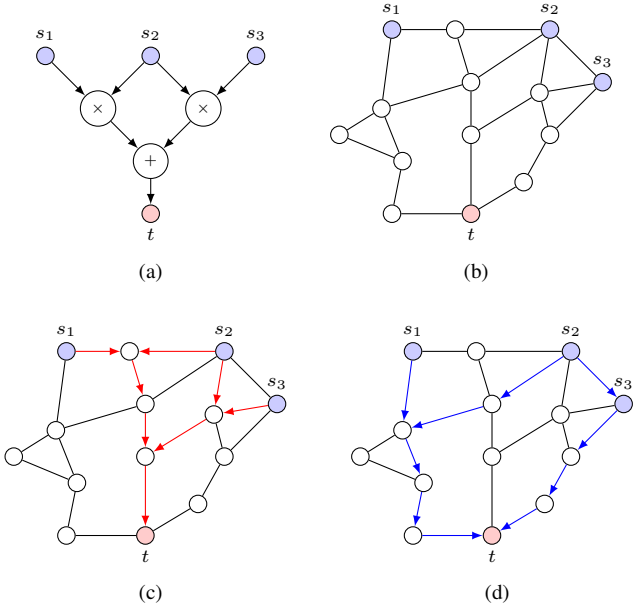
\begin{figure}[tbp]
  \centering
\subfloat[]{
\begin{tikzpicture}[>=latex,scale=0.7]
 \scriptsize
 \tikzstyle{every node} = [circle,draw=black]
 \node (a) at (-2,0) [fill=blue!20] {};
 \node (b) at (0,0) [fill=blue!20] {};
 \node (c) at (2,0) [fill=blue!20] {};
 \node (d) at (-1,-1)[font=\tiny]  {$\times$};
 \node (e) at (1,-1) [font=\tiny]{$\times$};
 \node (f) at (0,-2)[font=\tiny] {$+$};
 \node (g) at (0,-3) [fill=red!20]{};
 \node at (-2,0.4) [draw=none]{$s_1$};
 \node at (0,0.4) [draw=none]{$s_2$};
 \node at (2,0.4) [draw=none]{$s_3$};
 \node at (0,-3.4) [draw=none]{$t$};
 \draw [->] (a) -- (d) ;
 \draw [->] (b) -- (d) ;
 \draw [->] (b) -- (e) ;
 \draw [->] (c) -- (e) ;
 \draw [->] (d) -- (f) ;
 \draw [->] (e) -- (f) ;
 \draw [->] (f) -- (g) ;
 
 \end{tikzpicture}
}
\hspace*{10pt}
\subfloat[]{
\begin{tikzpicture}[>=latex,scale=0.7]
  \scriptsize
  \tikzstyle{every node} = [circle,draw=black]
  \node (a) at (-2,0) [fill=blue!20]  {};
  \node (b) at (-0.8,0) {};
  \node (c) at (1,0)  [fill=blue!20]  {};
  \node (d) at (-2.2,-1.5) {};
  \node (e) at (-0.5,-1) {};
  \node (f) at (0.8,-1.2) {};
  \node (g) at (1,-2)  {};
  \node (h) at (0.5,-2.9) {};
  \node (i) at (-0.5,-2) {};
  \node (j) at (-1.8,-2.5) {};
  \node (k) at (-0.5,-3.5) [fill=red!20] {};
  \node (l) at (-2,-3.5){};
  \node (m) at (2,-1)[fill=blue!20] {};
  \node (n) at (-3,-2){};
  \draw [-] (a) -- (b) -- (c) -- (m) --(f) --(c) -- (e)  --(i) -- (k);
  \draw [-] (a) -- (d) -- (n) -- (j) -- (l) --(k);
  \draw [-] (m) -- (g) -- (h) --(k);
  \draw [-] (e) -- (d) -- (j);
  \draw [-] (i) -- (f) -- (g);
  \draw [-] (b) -- (e);
  \node at (a) [draw=none,above] {$s_1$};
  \node at (c) [draw=none,above] {$s_2$};
  \node at (m) [draw=none,above] {$s_3$};
  \node at (-0.5,-3.9) [draw=none] {$t$};
 \end{tikzpicture}
} 

\subfloat[]{
\begin{tikzpicture}[>=latex,scale=0.7]
  \scriptsize
  \tikzstyle{every node} = [circle,draw=black]
  \node (a) at (-2,0) [fill=blue!20]  {};
  \node (b) at (-0.8,0) {};
  \node (c) at (1,0)  [fill=blue!20]  {};
  \node (d) at (-2.2,-1.5) {};
  \node (e) at (-0.5,-1) {};
  \node (f) at (0.8,-1.2) {};
  \node (g) at (1,-2)  {};
  \node (h) at (0.5,-2.9) {};
  \node (i) at (-0.5,-2) {};
  \node (j) at (-1.8,-2.5) {};
  \node (k) at (-0.5,-3.5) [fill=red!20] {};
  \node (l) at (-2,-3.5){};
  \node (m) at (2,-1) [fill=blue!20] {};
  \node (n) at (-3,-2) {};
  \draw [-] (c) -- (e)  ;
  \draw [-] (a) -- (d) -- (n) -- (j) -- (l) --(k);
  \draw [-] (m) -- (g) -- (h) --(k);
  \draw [-] (e) -- (d) -- (j);
  \draw [-] (f) -- (g);
  \draw [-] (c) -- (m) ;
  \draw [->,color=red] (a) -- (b)  ;
  \draw [->,color=red] (b) -- (e) ;
  \draw [->,color=red] (e) -- (i) ;
  \draw [->,color=red] (i) -- (k) ;
  \draw [->,color=red] (m) -- (f) ;
  \draw [->,color=red] (c) -- (b) ;
  \draw [->,color=red] (c) -- (f) ;
  \draw [->,color=red] (f) -- (i) ;
  \node at (a) [draw=none,above] {$s_1$};
  \node at (c) [draw=none,above] {$s_2$};
  \node at (m) [draw=none,above] {$s_3$};
  \node at (-0.5,-3.9) [draw=none] {$t$};
\end{tikzpicture}
} 
\hspace*{10pt}
\subfloat[]{
\begin{tikzpicture}[>=latex,scale=0.7]
  \scriptsize
  \tikzstyle{every node} = [circle,draw=black]
  \node (a) at (-2,0) [fill=blue!20]  {};
  \node (b) at (-0.8,0) {};
  \node (c) at (1,0)  [fill=blue!20]  {};
  \node (d) at (-2.2,-1.5) {};
  \node (e) at (-0.5,-1) {};
  \node (f) at (0.8,-1.2) {};
  \node (g) at (1,-2)  {};
  \node (h) at (0.5,-2.9) {};
  \node (i) at (-0.5,-2) {};
  \node (j) at (-1.8,-2.5) {};
  \node (k) at (-0.5,-3.5) [fill=red!20] {};
  \node (l) at (-2,-3.5){};
  \node (m) at (2,-1) [fill=blue!20] {};
  \node (n) at (-3,-2) {};
  \draw [-] (a) -- (b) -- (c) --(f) -- (m);
  \draw [-]  (d) -- (n) ;
  \draw [-] (i) -- (f) -- (g);
  \draw [-] (b) -- (e);
  \draw [-] (e)  --(i) -- (k);
  \draw [-] (n) -- (j);
  \draw [->,color=blue] (a) -- (d)  ;
  \draw [->,color=blue] (d) -- (j);
  \draw [->,color=blue] (j) -- (l) ;
  \draw [->,color=blue] (l) -- (k) ;
  \draw [->,color=blue] (c) -- (e)  ;
  \draw [->,color=blue] (e) -- (d) ;
  \draw [->,color=blue] (c) -- (m) ;
  \draw [->,color=blue] (m) -- (g) ;
  \draw [->,color=blue] (g) -- (h) ;
  \draw [->,color=blue] (h) -- (k) ;
  \node at (a) [draw=none,above] {$s_1$};
  \node at (c) [draw=none,above] {$s_2$};
  \node at (m) [draw=none,above] {$s_3$};
  \node at (-0.5,-3.9) [draw=none] {$t$};
\end{tikzpicture}
} 
  \caption{Computation of function $r(t) = \frac{x_1x_2+x_2x_3}{2}$.
    (a) A schema to compute $r(t)$ (b) Communication Network (c) Implementation 1 (d) Implementation 2}
  \label{fig:r-of-t}
\end{figure}

Rather than arithmetic operations to perform a function computation,
the nodes in Fig.~\ref{fig:r-of-t}a could correspond to operations
required to execute a database query and the directed edges would
represent the flow of the results of the operations. In this case the
graph of Fig.~\ref{fig:r-of-t}a is called a query graph. The input
data to this query graph could be from a distributed database, which
in turn could be a sensor network; in this case the graph in
Fig.~\ref{fig:r-of-t}b would represent the interconnection among the
units of the distributed database and some other nodes that may be
available for computation. In this context performing an efficient
query requires that the `operator placement' be efficient. Efficient
operator placement is addressed in
\cite{Bonfils04,Srivastsava05,Abrams05,Pietzuch06,Ying08,Wen13} all of
which assume that the query graph is a tree. While
\cite{Wen13,Pietzuch06,Bonfils04} develop heuristics based algorithms
for efficient operator placement when the query graph is a tree,
algorithms with formal analyses are available in
\cite{Srivastsava05,Abrams05,Ying08}. Although \cite{Phatak10}
considers a non-tree query graph, only heuristic algorithms are
provided.

Going back in the literature, we see that the `module placement'
problem for distributed processing from the 1980s has a flavor similar
to in-network computation and operator placement problems; see e.g.,
\cite{Stone77,Bokhari81,MaryLo88,Bokhari88,Fernandez-Baca89}. In this
case the nodes of Fig.~\ref{fig:r-of-t}a correspond to the modules of
a program and a directed edge $(u,v)$ implies that module $u$ calls
module $v$ during execution.  Further, the graph of
Fig.~\ref{fig:r-of-t}a is called a \textit{call graph.} The nodes of
Fig.~\ref{fig:r-of-t}b represent the processors on which the modules
of the program need to be executed and the edges represent the inter
processor links. The cost structure for this problem is more
complex. The execution cost is specified for each module-processor
pair and there is also an inter processor communication cost over an
edge that in turn depends on the modules placed at the ends of the
edge. The objective is to place the modules on the processors such
that the total cost of execution is minimized.  The first centralized
algorithm for optimal module placement when the call graph is a
directed tree was given in \cite{Bokhari81} and that for $k$-trees was
given in \cite{Fernandez-Baca89}. \cite{Stone77} showed that the
problem can be efficiently solved for a two processor system by using
network flow algorithms and \cite{MaryLo88} uses a similar approach to
develop heuristic algorithms for a general call graph.

From the preceding discussion we see that the objectives of, and hence
the solution techniques for, efficient in-network computation,
operator placement and module placement all have a very similar
theme---embedding (a formal definition is provided in
Section~\ref{sec:prelims}) a graph representing a computation schema
on a connected weighted graph representing a network of
processors. Much of the literature on such problems assumes that the
graph describing the function is a tree. This is clearly very
restrictive because the computation schema for a very large class of
useful computable functions cannot be described by a tree and are more
generally described by a DAG, e.g., fast Fourier transform (FFT),
sorting, any polynomial function of input data, and any function of
Boolean data.  MapReduce, the popular cloud computing paradigm, also
has a DAG representation and is discussed in some detail in
Example~\ref{ex:mapreduce} in Section~\ref{sec:layered-graph}.

\subsection{Organization of the Paper}

Our interest is in computing arbitrary functions that have a specific
algorithmic representation and the communication network is an
arbitrary network and not a random wireless network. (We reiterate
that although we assume the in-network computation scenario to anchor
the discussion, our results are also directly applicable to the
operator and module placement scenarios.)  Further, we do not seek to
specifically maximize the computation throughput; rather, our interest
is in minimizing the cost or delay in a `one shot' computation of the
function. As we mentioned before, two measures of efficiency are used
in this paper---cost of computation and delay in computation. Our
results on minimum cost embedding can be used to maximize computation
throughput when used in the algorithms developed in \cite{Shah13}.

Given an arbitrary function via its DAG description and an arbitrary
network over which this function is to be computed, our first interest
is to analyze the complexity of finding the optimal computation and
communication scheme that will compute the function in the
network. While much of the literature claims that this problem is
NP-complete (for the case of minimizing cost), to the best of our
knowledge a formal proof is not available; the literature eventually
leads to a citation to a private communication in \cite{Bokhari81}. In
Section~\ref{sec:hardness} we prove that in general, both the minimum
cost and the minimum delay embedding problems are NP-complete.

Some structure in the DAG can be exploited to provide polynomial time
algorithms for our problems. If the DAG is a tree, which is the
assumption in much of the extant literature, the minimum cost
embedding is similar to shortest path algorithms
\cite{Ying08,Shah13}. To the best of our knowledge, minimum delay
embeddings are not considered in the literature and in
Section~\ref{sec:tree} we provide a polynomial time algorithm to find
the minimum delay embedding when the computation schema is a tree.

Next we consider two large classes of computation graphs---(1)~layered
graphs, (2)~bounded treewidth graphs. We derive the motivation for
layered graphs from distributed data processing frameworks like
MapReduce \cite{Dean04} and Dryad \cite{Isard07}.  In
Section~\ref{sec:layered-graph} we provide a polynomial time algorithm
to find the minimum cost embedding when the DAG is a layered
graph. This same algorithm also obtains an approximation for the
minimum delay embedding.

In Section~\ref{sec:treewidth}, we show that the algorithm from
Section~\ref{sec:layered-graph} can find the minimum cost embedding
when the DAG is a bounded treewidth graph.  
The notation and the
formal problem definition is described in the next section.

In Section~\ref{sec:change_computation}, we describe an update
mechanism when there is perturbation in the DAG and we conclude with a
discussion in Section~\ref{sec:discuss}.

\section{Preliminaries}
\label{sec:prelims}

The communication network is represented by an undirected connected
graph $\net = (\netnodes,\netedges)$ with $\netnodes$ being the set of
$n$ nodes and $\netedges$ being the set of $m$ edges. The elements of
$\netnodes$ are denoted by $\{u_1, \ldots, u_n\}.$ Each edge
$(u_i,u_j) \in \netedges$ has a non negative weight $T(u_i,u_j)$
associated with it. The weight could, for example, correspond to
transmission time of a bit on the link, or the energy required to
transmit one bit on the link or something more abstract.  For a given
$T,$ and any $u_i,u_j \in \netnodes$ let $\distance_{u_i,u_j}$ be the
weight of the minimum cost path from $u_i$ to $u_j.$ Let $[[\mathtt{D}
= \distance_{u_i,u_j}]]$ be the $n \times n$ distance matrix.  Of the
$n$ nodes in $\net$, there are $K$ source nodes denoted by $\{s_1,s_2,
\ldots,s_K\} \subset \netnodes.$ Source node $s_i$ generates data
$x_i;$ denote $x=\{x_1,\ldots,x_K\}.$ A sink node ${t} \in \netnodes$
requires to obtain a function $f(x_1,x_2, \ldots, x_K)$ of the data.

We assume that schema to compute $f(x_1,x_2, \ldots, x_K)$ is given
and is represented by a directed acyclic graph $\compgraph =
(\compnodes,\compedges),$ where $\compnodes$ is the set of $p$ nodes
and $\compedges$ is the set of $q$ edges. The nodes in $\compgraph$
are denoted by $\{\omega_1, \ldots, \omega_p\}$ and correspond to
operations that need to be performed on the input data to the node and
the outgoing edges denote the flow of the result of these
operations. Thus each edge in $\compgraph$ represents a sub-function
of the inputs. The sources in the computation graph are denoted by
nodes $\{\omega_1,\omega_2,\ldots,\omega_K\}$ with node $\omega_i$
corresponding to source $x_i;$ node $\omega_p$ is the sink that
receives the function $f(x_1,x_2, \ldots, x_K).$

The direction on the edges in $\compgraph$ represent the direction of
the flow of the data.  Each edge $(\omega_i,\omega_j) \in \compedges$
has a non negative weight $\edgewt(\omega_i,\omega_j)$ associated with
it which could correspond to the number of bits used to represent the
intermediate function. 

Since $\compgraph$ is a directed acyclic graph there is a partial
order associated with its vertices. If $(\omega_i,\omega_j) \in
\compedges$ then the function at $\omega_j$ cannot be computed until
the function at $\omega_i$ is computed and the result forwarded to
$\omega_j.$ Let $\prenodes \omega$ and $\sucnodes \omega$ denote,
respectively, the immediate predecessors and successors of vertex
$\omega,$ i.e., $\prenodes \omega = \{\tau \in \compnodes | (\tau,\omega) \in \compedges\}$ and $\sucnodes \omega = \{\tau \in \compnodes | (\omega,\tau) \in \compedges\}.$
A processing cost (delay) function $\processingwt: \compnodes \times
\netnodes \mapsto \mathbb{R}^+$ is used to capture the cost (delay) of
performing a particular operation on a particular vertex of the
network. Now we define the embedding of $\compgraph$ on $\net$ as
follows.
\begin{definition}
  \label{def:embedding}
  An embedding of a computation graph $\compgraph$ on a communication
  network $\net$ is a many-to-one function $\embedding: \compnodes
  \mapsto \netnodes$ which satisfies the following conditions. 
  \begin{enumerate}
  \item $\embedding(\omega_i) = s_i$ for $i=\{1,\ldots,K\}$
  \item $\embedding(\omega_p) = t.$
  \end{enumerate}
\end{definition}
In this definition of embedding each node in the computation graph is
mapped to a single node in the network graph and the edge
$(\omega_i,\omega_j) \in \compedges$ is mapped to the shortest path
between $\embedding(\omega_i)$ and $\embedding(\omega_j).$ Alternate
definitions of an embedding are possible, e.g., an edge in
$\compgraph$ can be mapped to more than one path in $\net$ while
satisfying some continuity constraints; this is the definition of an
embedding in \cite{Shah13}. 

An embedding defines a computation and communication
sequence in $\net$ to obtain $f$ at the sink. Let $\setofembeddings$
be the set of all embeddings of $\compgraph$ in $\net$ which
follow Definition~\ref{def:embedding}.  The weight functions $T$ and
$\processingwt$ can be treated as, respectively, the communication and the
processing delays in $\net$ for computing and communicating the
sub-functions leading to computation of $f.$ We can define the delay in
computing a sub-function by a node $\omega \in \compnodes$ in the
embedding $\embedding$ as
\begin{equation}
   \begin{split}
    d(\embedding(\omega_i)) := & \max_{\omega_j \in \prenodes
      {\omega_i}} [d(\embedding(\omega_j)) +
      \edgewt(\omega_j,\omega_i)\distance_{\embedding(\omega_j),
        \embedding(\omega_i)}] + \\ &
    \processingwt(\omega_i,\embedding(\omega_i)). \label{eq:nodedelay}
  \end{split} 
\end{equation} 
Recall that $\distance_{u,v}$ is the length of the shortest path between
vertices $u,v \in \netnodes;$ thus the first term here corresponds to
the delay in obtaining all the operands at node $\embedding(\omega_i)$
and the second term corresponds to the processing delay at the
node. Setting the delay at the sources to zero, i.e.,
$d(\embedding(\omega_i)) = 0$ for all $i \in [1,K],$ we can
recursively calculate the delay of each vertices of $\compgraph$ on
$\net.$ The delay $d(\embedding)$ of an embedding $\embedding$ is defined
 as the delay of the sink, i.e., $d(\embedding) := d(\embedding(\omega_p)).$

This leads us to the first problem that we consider in this
paper:~Find an embedding $\embedding_{opt}^d$ such that the delay of
the embedding is minimum among all the embeddings for a given
$\compgraph, \net, T,\edgewt$ and $\processingwt,$ i.e., solve the
optimization problem
\begin{equation}
  \label{mindelay}
  \tag{\textsf{MinDelay}}
  \embedding_{opt}^d := \argmin_{\embedding \in \setofembeddings}
  d(\embedding).
\end{equation}

Alternatively, considering the weight functions $T,$ $\edgewt,$ and
$\processingwt$ as cost of communication and computation, e.g., the
energy cost, the cost of an embedding can be defined as
\begin{equation}
  C(\embedding) := \sum_{\omega \in \compnodes}
  \processingwt(\omega,\embedding(\omega)) + \sum_{(\omega_i,\omega_j) \in
    \compedges} \left(\edgewt(\omega_i,\omega_j)
  \distance_{\embedding(\omega_i)\embedding(\omega_j)}\right).
  \label{eq:embedding-cost}
\end{equation}
We can then find an embedding $\embedding_{opt}^c$ such that the cost
of the embedding is minimum among all the embeddings for a given
$\compgraph, \net, T,\edgewt$ and $\processingwt,$ i.e., solve the
optimization problem
\begin{equation}
  \label{mincost}
  \tag{\textsf{MinCost}}
  \embedding_{opt}^c := \argmin_{\embedding \in \setofembeddings} C(\embedding).
\end{equation}

The following example illustrates the preceding problems and the
system set up. 

\begin{example}
 \label{ex:delaycostcomp}
  Consider a computation graph $\compgraph=(\compnodes,\compedges)$
  and communication network $\net=(\netnodes,\netedges)$ shown in
  Fig.~\ref{fig:sysmodel}.  The labels of each vertex in both the
  graphs are shown in Fig.~\ref{fig:sysmodel}.  The processing
  cost (delay) for sources is assumed to be zero and for other
  vertices of $\compgraph$ it is assumed to be unity.
  \begin{figure}[tbp]
   \centering
\subfloat[]{
\begin{tikzpicture}[>=latex,scale=0.7]
 \scriptsize
 \tikzstyle{every node} = [circle,draw=black]
 \node (a) at (-2,0) [fill=blue!20] {};
 \node (b) at (0,0) [fill=blue!20] {};
 \node (c) at (2,0) [fill=blue!20] {};
 \node (d) at (-1,-1)[font=\tiny]  {$+$};
 \node (e) at (1,-1) [font=\tiny]{$+$};
 \node (f) at (0,-2)[font=\tiny] {$\times$};
 \node (g) at (0,-3) [fill=red!20]{};
 \node at (-2,0.4) [draw=none]{$\omega_1$};
 \node at (0,0.4) [draw=none]{$\omega_2$};
 \node at (2,0.4) [draw=none]{$\omega_3$};
 \node at (0,-3.4) [draw=none]{$\omega_7$};
 \node at (d) [draw=none,right=4pt] {$\omega_4$};
 \node at (e) [draw=none,right=4pt] {$\omega_5$};
 \node at (f) [draw=none,right=4pt] {$\omega_6$};
 \draw [->] (a) -- (d) ;
 \draw [->] (b) -- (d) ;
 \draw [->] (b) -- (e) ;
 \draw [->] (c) -- (e) ;
 \draw [->] (d) -- (f) ;
 \draw [->] (e) -- (f) ;
 \draw [->] (f) -- (g) ;
\end{tikzpicture}
}
\subfloat[]{
\begin{tikzpicture}[>=latex,scale=0.7]
  \scriptsize
  \tikzstyle{every node} = [circle,draw=black]
 \node (a) at (-2,0) [fill=blue!20] {};
 \node (b) at (0,0) [fill=blue!20] {};
 \node (c) at (2,0) [fill=blue!20] {};
 \node (d) at (-1,-1)  {};
 \node (e) at (1,-1) {};
 \node (f) at (0,-1) {};
 \node (g) at (0,-2) {};
 \node (h) at (0,-3) [fill=red!20]{};
 \node at (-2,0.4) [draw=none]{$s_1$};
 \node at (0,0.4) [draw=none]{$s_2$};
 \node at (2,0.4) [draw=none]{$s_3$};
 \node at (0,-3.4) [draw=none]{$t$};
 \node at (d) [draw=none,right] {$a$};
 \node at (f) [draw=none,right] {$b$};
 \node at (e) [draw=none,right] {$c$};
 \node at (g) [draw=none,right] {$d$};
 \draw [-] (a) -- (d)node [draw=none, midway,left] {$10$};
 \draw [-] (d) -- (g)node [draw=none, midway,left] {$1$} ;
 \draw [-] (b) -- (d) node [draw=none, midway,left] {$2$};
 \draw [-] (c) -- (f)node [draw=none, pos=0.05,left] {$12$};
 \draw [-] (b) -- (e)node [draw=none, pos=0.2,right] {$8$}; 
 \draw [-] (e) -- (g)node [draw=none, midway,right] {$1$};
 \draw [-] (c) -- (e) node [draw=none, midway,right] {$10$};
 \node at (-0.2,-0.5) [draw=none] {$4$};
 \draw [-] (b) -- (f);
 \node at (-0.2,-1.5) [draw=none] {$1$};
 \draw [-] (f) -- (g);
 \node at (-0.2,-2.5) [draw=none] {$1$};
 \draw [-] (g) -- (h);
\end{tikzpicture}
} 
\vspace*{6pt}
\subfloat[]{
\begin{tikzpicture}[>=latex,scale=0.7]
  \scriptsize
  \tikzstyle{every node} = [circle,draw=black]
 \node (a) at (-2,0) [fill=blue!20] {};
 \node (b) at (0,0) [fill=blue!20] {};
 \node (c) at (2,0) [fill=blue!20] {};
 \node (d) at (-1,-1) [fill=yellow!20] {};
 \node (e) at (1,-1) [fill=yellow!20] {};
 \node (f) at (0,-1) {};
 \node (g) at (0,-2) [fill=yellow!20] {};
 \node (h) at (0,-3) [fill=red!20]{};
 \node at (-2,0.4) [draw=none]{$s_1$};
 \node at (0,0.4) [draw=none]{$s_2$};
 \node at (2,0.4) [draw=none]{$s_3$};
 \node at (0,-3.4) [draw=none]{$t$};
 \node at (-0.5,-1) [draw=none] {$\omega_4$};
\node at (0.4,-1) [draw=none] {$b$};
\node at (1.5,-1) [draw=none] {$\omega_5$};
 \node at (0.5,-2) [draw=none] {$\omega_6$};
 \draw [->] (a) -- (d);
 \draw [->] (b) -- (d);
 \draw [->] (b) -- (e);
 \draw [->] (c) -- (e);
 \draw [->] (d) -- (g);
 \draw [->] (e) -- (g);
 \draw [->] (g) -- (h);
\end{tikzpicture}
} 
\subfloat[]{
\begin{tikzpicture}[>=latex,scale=0.7]
  \scriptsize
  \tikzstyle{every node} = [circle,draw=black]
 \node (a) at (-2,0) [fill=blue!20] {};
 \node (b) at (0,0) [fill=blue!20] {};
 \node (c) at (2,0) [fill=blue!20] {};
 \node (d) at (-1,-1) [fill=yellow!20] {};
 \node (e) at (1,-1) {};
 \node (f) at (0,-1) [fill=yellow!20] {};
 \node (g) at (0,-2)[fill=yellow!20] {};
 \node (h) at (0,-3) [fill=red!20]{};
 \node at (-2,0.4) [draw=none]{$s_1$};
 \node at (0,0.4) [draw=none]{$s_2$};
 \node at (2,0.4) [draw=none]{$s_3$};
 \node at (0,-3.4) [draw=none]{$t$};
 \node at (-0.5,-1) [draw=none] {$\omega_4$};
\node at (0.4,-1.1) [draw=none] {$\omega_5$};
 \node at (1.4,-1) [draw=none] {$c$};
\node at (0.5,-2) [draw=none] {$\omega_6$};
 \draw [->] (a) -- (d);
 \draw [->] (b) -- (d);
 \draw [->] (d) -- (g);
 \draw [->] (g) -- (h);
 \draw [->] (b) -- (f);
 \draw [->] (f) -- (g);
 \draw [->] (c) -- (f);
\end{tikzpicture}
} 
    \caption{(a). Computation graph $\compgraph$ for $f =
      (x_1+x_2)(x_2+x_3)$} (b). Communication Network $\net.$ The
    numbers near the edges represent the weights on them.
    (c). Minimum delay embedding $\embedding_1.$ 
    (d). Minimum cost embedding
    $\embedding_2.$
    \label{fig:sysmodel}
  \end{figure}
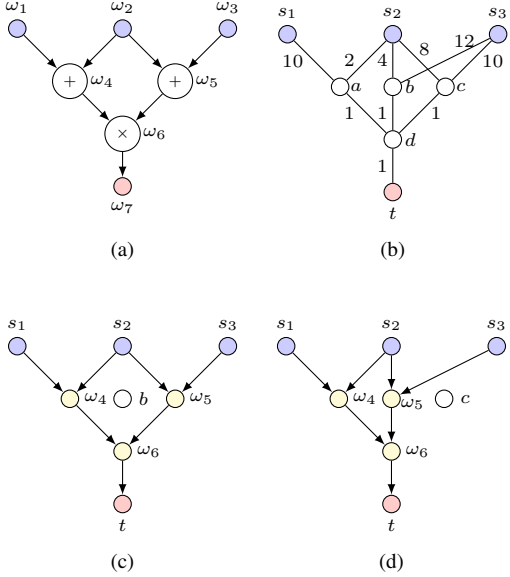
  An embedding $\embedding$ of $\compgraph$ on $\net$ will have
  $\embedding(\omega_i) = s_i \ \ \forall i \in [1,3]$ and
  $\embedding(\omega_7) = t.$ Now consider two embeddings
  $\embedding_1,\embedding_2$ such that $\embedding_1(\omega_4) =
  a,\embedding_1(\omega_5) = c,\embedding_1(\omega_6) = d$ and
  $\embedding_2(\omega_4) = a,\embedding_2(\omega_5) =
  b,\embedding_2(\omega_6) = d.$ These are shown in
  Figs.~\ref{fig:sysmodel}c~and~\ref{fig:sysmodel}d.

  Using \eqref{eq:embedding-cost}, it is easy to verify that
  $C(\embedding_1) = 36$ and $C(\embedding_2) = 34.$ The delays in the
  embedding $\embedding_1$ are: $d(\embedding_1(\omega_4)) =
  \max(10,2) +1= 11, d(\embedding_1(\omega_5)) = \max(8,10) +1= 11,
  d(\embedding_1(\omega_6)) = \max(12,12) +1= 13$ and finally
  $d(\embedding_1) = d(\embedding_1(\omega_7)) = 13+1=14.$ Similarly,
  $d(\embedding_2) = 16.$ Observe that the delay of $\embedding_1$ is
  lower among the two but its cost is higher than that of
  $\embedding_2.$ 
\end{example}

Now we present an example which shows that the difference between the
delay obtained by the solution of \mincost\ problem and that of the
\mindelay\ problem can be of the order of the number of sources in the
network.
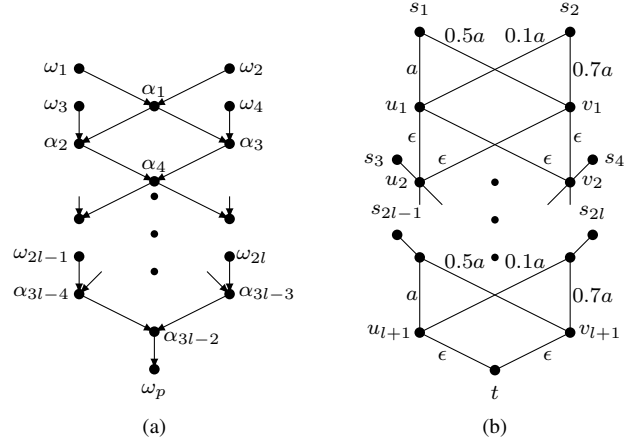
\begin{figure}[tbp]
\subfloat[]{
\begin{tikzpicture}[>=latex]
\footnotesize
 \foreach \x in {-1,1}{
  \foreach \y in {2,1.5,1,-0.5,0,-1}{
   \fill (\x,\y) circle (0.07cm);
   }
 }
 \foreach \y in {1.5,0.5}{
  \fill (0,\y) circle (0.07cm);
  \draw [->] (0,\y) -- (-1,\y-0.5);
  \draw [->] (0,\y) -- (1,\y-0.5);
  }  
 \foreach \x in {-1,1}{
  \draw [->] (\x,2) -- (0,1.5);
  \draw [->] (\x,1) -- (0,0.5);
  \draw [->] (\x,1.5) -- (\x,1);
  \draw [->] (\x,-0.5) -- (\x,-1);
  \draw [->] (\x,-1) -- (0,-1.5);
  }
\foreach \y in {0.3,-0.2,-0.7} {
 \fill (0,\y) circle (0.05cm);
 }
\fill (0,-1.5) circle (0.07cm);
\fill (0,-2) circle (0.07cm);
\draw [->] (0,-1.5) -- (0,-2);
 \draw [->] (-0.7,-0.7) --(-1,-1);
 \draw [->] (0.7,-0.7) -- (1,-1);
 \draw [->] (-1,0.3) -- (-1,0);
 \draw[->] (1,0.3) -- (1,0);
 
 \node at (-1.3,2) [draw=none] {$\omega_1$};
 \node at (1.3,2) [draw=none] {$\omega_2$};
 \node at (-1.3,1.5) [draw=none] {$\omega_3$};
 \node at (1.3,1.5) [draw=none] {$\omega_4$}; 
 \node at (-1.5,-0.5) [draw=none] {$\omega_{2l-1}$};
 \node at (1.3,-0.5) [draw=none] {$\omega_{2l}$};
 \node at (0,-2.3) [draw=none] {$\omega_p$};
 \node at (-1.3,1) [draw=none] {$\alpha_2$};
 \node at (1.3,1) [draw=none] {$\alpha_3$};
 \node at (0,1.7) [draw=none] {$\alpha_1$};
 \node at (0,0.7) [draw=none] {$\alpha_4$};
 \node at (-1.5,-1) [draw=none] {$\alpha_{3l-4}$};
 \node at (1.5,-1) [draw=none] {$\alpha_{3l-3}$};
 \node at (0.5,-1.6) [draw=none] {$\alpha_{3l-2}$};
 \end{tikzpicture}
}
\hspace*{10pt}
  \subfloat[]{
\begin{tikzpicture}[>=latex]
\footnotesize
 \foreach \x in {-1,1} {
  \foreach \y in {2,1,0} {
   \fill (\x,\y) circle (0.07cm);
  }
 }
 \foreach \x in {-1,1} {
  \foreach \y in {2,1,-1} {
    \draw [-] (\x,\y) -- (\x,\y-1); 
   }
 }
 \draw [-] (-1,0) -- (-1,-0.3);
 \draw [-] (1,0) -- (1,-0.3);
 \draw [-] (-1,0) -- (-0.7,-0.3);
 \draw [-] (1,0) -- (0.7,-0.3);
 
  \foreach \y in {2,1,-1} {
   \draw [-] (-1,\y) -- (1,\y -1);
   \draw [-] (1,\y) -- (-1,\y-1);
   } 
  \foreach \y in {-1,0,-0.5} {
   \fill (0,\y) circle (0.05cm);
   } 
   \fill (-1.3,0.3) circle (0.07cm);
   \fill (1.3,0.3) circle (0.07cm);
   \draw [-] (-1.3,0.3) -- (-1,0);
   \draw [-] (1.3,0.3) -- (1,0);
   \fill (-1.3,-0.7) circle (0.07cm);
   \fill (1.3,-0.7) circle (0.07cm);
   \draw [-] (-1.3,-0.7) -- (-1,-1);
   \draw [-] (1.3,-0.7) -- (1,-1);
   \node at (-1,2.3) [draw=none] {$s_1$};
   \node at (1,2.3) [draw=none] {$s_2$};
   \node at (-1.6,0.3) [draw=none] {$s_3$};
   \node at (1.6,0.3) [draw=none] {$s_4$};
   \node at (-1.3,-0.4) [draw=none] {$s_{2l-1}$};
   \node at (1.3,-0.4) [draw=none] {$s_{2l}$};
   \node at (0,-2.8) [draw=none] {$t$};
   \foreach \x in {-1,1} {
    \foreach \y in {-1,-2} {
     \fill (\x,\y) circle (0.07cm);
     }
     }
     \draw [-] (-1,-2) -- (0,-2.5);
     \draw [-] (1,-2) -- (0,-2.5);
     \fill (0,-2.5) circle (0.07cm);
     \node at (-1.1,1.5) [draw=none] {$a$};
     \node at (1.3,1.5) [draw=none] {$0.7a$};
     \node at (-0.4,1.95) [draw=none] {$0.5a$};
     \node at (0.4,1.95) [draw=none] {$0.1a$};
     \node at (-1.1,-1.5) [draw=none] {$a$};
     \node at (1.3,-1.5) [draw=none] {$0.7a$};
     \node at (-0.4,-1.05) [draw=none] {$0.5a$};
     \node at (0.4,-1.05) [draw=none] {$0.1a$};
     \node at (-0.7,0.3) [draw=none] {$\epsilon$};
     \node at (0.7,0.3) [draw=none] {$\epsilon$};
     \node at (-1.1,0.6) [draw=none] {$\epsilon$};
     \node at (1.1,0.6) [draw=none] {$\epsilon$};
     \node at (-0.7,-2.3) [draw=none] {$\epsilon$};
     \node at (0.7,-2.3) [draw=none] {$\epsilon$};
     \node at (-1.3,1) [draw=none] {$u_1$};
     \node at (1.3,1) [draw=none] {$v_1$};
     \node at (-1.3,0) [draw=none] {$u_2$};
     \node at (1.3,0) [draw=none] {$v_2$};
     \node at (-1.4,-2) [draw=none] {$u_{l+1}$};
     \node at (1.4,-2) [draw=none] {$v_{l+1}$};
 \end{tikzpicture}
}
  \caption{Illustrating Example~\ref{ex:compare}. (a). Computation graph with $2l$
    sources and a sink.  (b). Communication network.}
  \label{fig:delaycost_diff}
\end{figure}
\begin{example}
  \label{ex:compare}
  Consider the computation graph $\compgraph =
  (\compnodes,\compedges)$ and a network graph
  $\net=(\netnodes,\netedges)$ as shown in
  Figs.~\ref{fig:delaycost_diff}a and b respectively. The labels of
  the vertices are shown in the figure and the numbers near the edges
  of $\net$ represent the weight of that edge. Note that the structure
  between $s_1$ to $u_2$ and $s_2$ to $v_2$ is repeated in the network
  graph $l$ times. Let us call this structure as $A.$ We assume that
  the weights on the edges of $\compgraph$ are all one and the
  processing costs are zero for sources and for all other vertices it
  is assumed unity. Any embedding $\embedding$ of $\compgraph$ on
  $\net$ will have $\embedding(\omega_i) = s_i \ \ \forall i \in
  [1,2l]$ and $\embedding(\omega_p) = t.$ Consider an embedding
  $\embedding_1$ such that $\embedding_1(\alpha_1) =
  u_1,\embedding_1(\alpha_2) = u_2, \embedding_1(\alpha_3) = v_2,
  \ldots, \embedding(\alpha_{3l-2}) = u_{l+1}.$ Let us compute the
  cost of this embedding. Note that the total cost of the embedding is
  $l$ times the cost coming from the structure $A$ plus the weight of
  edge $u_{l+1}t.$ The cost due to embedding of $A$ is the sum of the
  weights of edges $s_1u_1,s_2u_1, u_1u_2,u_1v_2$ and the processing
  costs at $u_1,u_2,v_2.$ Hence the cost of the embedding is
  $C(\embedding_1) = l(3+1.1a+2\epsilon) +\epsilon.$ Similarly the
  delay of this embedding is $d(\embedding_1) = (a+\epsilon+2)l
  +\epsilon.$ Now consider another embedding $\embedding_2$ such that
  $\embedding_2(\alpha_1) = v_1,\embedding_2(\alpha_2) = u_2,
  \embedding_2(\alpha_3) = v_2, \ldots, \embedding_2(\alpha_{3l-2}) =
  v_{l+1}.$ The cost and delay for this embedding can be computed in a
  similar fashion to obtain $C(\embedding_2) = l(3+1.2a+2\epsilon)
  +\epsilon$ and $d(\embedding_2) = l(0.7a +\epsilon +2)l +\epsilon.$
  It can be shown that the first embedding $\embedding_1$ is the
  solution of \mincost\ where as $\embedding_2$ is the solution of
  \mindelay\ problem of $\compgraph$ on $\net.$ If we use the solution
  of \mincost\ to get the solution of \mindelay\ then the difference
  would be $d(\embedding_1)-d(\embedding_2) = 0.3al$ which is of the
  order of the number of sources in the graph.
\end{example}


\section{Hardness of Embedding}
\label{sec:hardness}

We begin by considering \mindelay. First consider the case when there
is no processing delay, i.e.,
$\processingwt(\omega,\embedding(\omega)) = 0$ in the network and the
computation graph is unweighted, i.e., $\edgewt(\omega_i,\omega_j) =1
\ \ \forall (\omega_i,\omega_j) \in \compedges.$
 
The delay of the embedding in this case is the delay of the longest
embedded path from any source to the sink in $\net.$ Let
$\distance_{s_it}$ be the delay of the minimum delay path between
source $s_i$ and the sink $t$ in $\net.$ Then the delay of any
embedding of $\compgraph$ on $\net$ which follows the conditions of
Definition~\ref{def:embedding} has to be more than the delay on the
longest of all the minimum delay paths from sources to sink. In other
words, $d(\embedding) \geq \max_{i \in [1,K]} (\distance_{s_it}).$
Now consider an embedding $\embedding^{*}$ which maps all the
intermediate vertices of $\compgraph$ to the sink in $\net.$ The delay
of this embedding will be $d(\embedding^{*}) = \max_{i \in [1,K]} (\distance_{s_it}).$
Comparing it with $d(\embedding)$ tells us that the embedding
$\embedding^{*}$ minimizes the delay. Hence \mindelay\ is easy to
solve if there are no processing delays and $\compgraph$ is
unweighted.

Now we analyze the hardness of \mindelay\ and \mincost\ for arbitrary
$\compgraph$ and $\net$ and show that the both the optimization
problems are NP-hard. We prove the hardness of the optimization
problems by proving that the corresponding decision versions are
NP-hard. The decision versions of the \mindelay\ and \mincost\ are
defined as follows:

\begin{definition}
  For a given $\compgraph,\net,T,\edgewt,\processingwt$ and a positive
  number $L$ the decision version of the \mindelay\ problem outputs
  \textit{``yes"} if there exists an embedding $\embedding$ of
  $\compgraph$ on $\net$ such that $d(\embedding) \leq L$ and outputs
  \textit{``no"} if no such embedding exists.
\end{definition}

\begin{definition} 
  For a given $\compgraph,\net,T,\edgewt,\processingwt$ and a positive
  number $L$ the decision version of the \mincost\ problem outputs
  \textit{``yes"} if there exists an embedding $\embedding$ of
  $\compgraph$ on $\net$ such that $C(\embedding) \leq L$ and outputs
  \textit{``no"} if no such embedding exists.
\end{definition}

Note that if one can solve the optimization version of \mindelay\
(resp. \mincost) problem in time, say $t,$ then using the solution of
that problem we can solve the corresponding decision version in time
$O(t).$ Hence the original optimization problem is atleast as hard as
its decision version. This implies that if we just prove that the
decision version of the \mindelay (resp. \mincost) is NP-hard then the
optimization version is also NP-hard. We now proceed to prove that the
decision version of our optimization problems are indeed NP-complete.

\begin{theorem}
\label{thm:delay}
The decision version of \mindelay\ problem is NP-complete.
\end{theorem}

\begin{proof}
  We first prove that the decision version of \mindelay\ is NP-hard by
  giving a reduction from the NP-complete problem \textit{Precedence
    Constraint Scheduling with fixed mapping (PCS-FM)}\cite{Garey79}.

  PCS-FM problem is defined as follows. Given a set of tasks $H$ with
  a partial order $\lessdot$ on it, task $h$ having length $l(h) = 1 \
  \ \forall h \in H,$ a set $\tau \subset H,$ and $m \in
  \mathbb{Z}^{+}$ processors, find a schedule $\sigma$ of tasks on the
  processors which meets an overall deadline $D^{\prime},$ maps each
  $\tau_i \in \tau$ to a particular processor $p(\tau_i)$ and obeys
  the precedence constraint that if $h_i \lessdot h_j$ then
  $\sigma(h_j) \geq \sigma(h_i) + 1.$ To the best of our knowledge the
  hardness of PCS-FM problem has not been proved in the literature and
  we provide the proof of NP-completeness of PCS-FM in Appendix~\ref{sec:PCS-FM}. 

  We first give a reduction from an instance $\phi =
  (H,\lessdot,\tau,p(\tau),m,l,D^{\prime})$ of PCS-FM to an instance
  of \mindelay\ $\psi =
  (\compgraph,\net,\gamma,\lambda,T,\processingwt,D),$ where
  $\compgraph = (\compnodes,\compedges)$ and $\net =
  (\netnodes,\netedges)$ are the computation and communication graphs
  respectively. The set $\gamma \subset \compnodes$ is to be mapped to
  $\lambda \subset \netnodes$ under any embedding $\embedding$ and
  $T,\processingwt$ are the communication and processing delay of the
  network. Note that the set of tasks $H$ along with the partial order
  creates a DAG $\lessdot$ and we define $\compgraph = \lessdot.$ We
  define $\net$ to be a complete graph on $m$ processors. Let $\gamma
  = \tau$ and $\lambda = p(\tau).$ The transmission delay $T$ between
  any two vertices of $\net$ is taken as $\epsilon = \frac{1}{|H|^2}$
  and $\processingwt(\omega,u) = 1$ for all $\omega \in \compnodes, u
  \in \netnodes.$ Finally $D=D^{\prime} +1.$

  We have to prove that there is a schedule $\sigma$ of $\phi$ which
  finishes in time $D^{\prime}$ if and only if there is an embedding
  $\embedding$ of $\psi$ with delay $D^{\prime} \leq D \leq
  D^{\prime}+1.$ The forward direction is easy to prove. If there is a
  schedule of $\phi$ which finishes in $D^{\prime}$ time and maps a
  task $h \in H$ to a processor $q$ then we can create an embedding of
  $\psi$ which maps the same vertex $h \in \compnodes$ to a vertex $q
  \in \netnodes.$ Note that because $\gamma = \tau$ and $\lambda =
  p(\tau)$ the conditions of Definition~\ref{def:embedding} are met in
  this embedding. The delay of any vertex $u \in \compgraph$ in this
  embedding will be at most the time at which the task $u$ finishes in
  the schedule $\sigma$ plus the number of times edges in $\net$ are
  used because of the same precedence order. Hence the delay of this
  embedding will be $D^{\prime} < D < D^{\prime}+ |H|\epsilon \leq
  D^{\prime} + \frac{1}{|H|} < D^{\prime} +1.$

  To complete the proof we need to prove that if there is an embedding
  $\embedding$ of $\psi$ of delay $\alpha \in \mathbb{R}^+$ then there
  is a schedule $\sigma$ of $\phi$ which finishes in time $\lfloor
  \alpha \rfloor.$ We will create a schedule $\sigma$ from the
  embedding $\embedding.$ If a vertex $h \in \compnodes$ is mapped to
  a vertex $u \in \netnodes$ then in the schedule $\sigma$ also the
  task $h$ is executed by processor $u.$ Because $\gamma = \tau$ and
  $\lambda = p(\tau)$ the tasks in $\tau$ are mapped to $p(\tau)$ in
  this schedule also. Let the total number of edge uses in this
  embedding be $b.$ Note that $b < |H|^2$ in any embedding where $|H|$
  are the number of vertices in graph $\compgraph.$ The total
  transmission delay in the embedding is $b \epsilon < 1.$ Recall that
  the processing delays are all $1$ hence the delay of the embedding
  can be written as $\alpha = \lfloor \alpha \rfloor + b \epsilon.$ It
  is easy to verify that the time required by the schedule $\sigma$ to
  complete is $\lfloor \alpha \rfloor.$ This proves that the decision
  of \mindelay\ is NP-hard.

  Given an instance of \mindelay\ problem an embedding $\embedding$
  can be guessed non-deterministically and checked whether
  $d(\embedding) \leq L$ in polynomial time. Thus the decision version
  of \mindelay\ problem is in NP and our reduction proves that it is
  in fact NP-complete.
\end{proof}

We now look at the problem of approximating the \mindelay\ problem. We
say that a polynomial time approximation algorithm has a performance
guarantee $\alpha >1$ if it outputs a feasible solution of the problem
which has delay at most $\alpha$ times the optimal solution.  We prove
the following:

\begin{theorem}
  \label{thm:delayapprox}
  Unless P=NP, an instance of the \mindelay\ problem with
  $(\compgraph,\net,T,\edgewt,\processingwt)$ and with unit processing
  delays and unit weights on the edges of $\compgraph,$ does not have
  a polynomial-time approximation algorithm with performance guarantee
  strictly less than $100/99$ if its solution is greater than $10.$
\end{theorem}

\begin{proof}
  Note that while proving the hardness of \mindelay\ we first reduced
  an instance of a PCS problem to an instance of a PCS-FM problem and
  then we reduced PCS-FM to \mindelay\ problem. Let $\mathsf{opt}_1$
  and $\mathsf{sol}_1$ be the deadline achieved by the optimal and a
  feasible solution of PCS problem respectively. Similarly, let
  $\mathsf{opt}_2$ and $\mathsf{sol}_2$ be the deadline achieved by
  the optimal and a feasible solution of the PCS-FM problem. Finally,
  let $\mathsf{opt}_3$ and $\mathsf{sol}_3$ be the optimal and
  feasible solutions for \mindelay. While proving the NP-completeness
  of PCS-FM (in Appendix~\ref{sec:PCS-FM}) we showed that $\mathsf{opt}_2 =
  \mathsf{opt}_1 +2.$ We also showed that if a schedule of PCS problem
  achieves a deadline $\mathsf{sol}_1,$ then there is a schedule of
  PCS-FM with $\mathsf{sol}_2 = \mathsf{sol}_1 +2.$ Let PCS-FM have a
  polynomial-time approximation algorithm with solution
  $\mathsf{sol}_2 = x \cdot \mathsf{opt}_2.$ By Observation~5.1 of
  \cite{Hoogeveen98} we know that the PCS problem does not have
  polynomial-time approximation algorithm with performance guarantee
  less than $4/3.$ Hence, $\mathsf{sol}_1 > \frac{4}{3} \mathsf{opt}_1.$
  Substituting the relation between $\mathsf{sol}_1$ and
  $\mathsf{sol}_2$ in the above equation we get, 
  $\mathsf{sol}_2 -2 > \frac{4}{3} \mathsf{opt}_1.$ Thus, 
  $ x \mathsf{opt}_2 > \frac{4}{3} \mathsf{opt}_1 +2$ and $ x > \frac{4}{3}   \frac{\mathsf{opt}_1}{\mathsf{opt}_2} +\frac{2}{\mathsf{opt}_2}.$
  We now observe that $\mathsf{opt}_1 \geq 1 \Rightarrow
  \mathsf{opt}_2 \geq 3$ and hence $\mathsf{opt}_1 >
  \frac{1}{3}\mathsf{opt}_2$ we can write, $x > \frac{4}{3} \frac{\mathsf{opt}_1}{\mathsf{opt}_2} + 1 - \frac{\mathsf{opt}_1}{\mathsf{opt}_2}$ and finally $x > 1 + \frac{1}{9}.$
  Hence, unless P=NP, PCS-FM cannot have a polynomial time
  approximation algorithm with performance guarantee less than $10/9.$
  While proving the hardness for \mindelay\ we showed that for any
  instance of PCS-FM with $\mathsf{sol}_2$ we can get an instance of
  \mindelay\ problem with solution $\mathsf{sol}_3$ such that
  $\mathsf{sol}_2 \leq \mathsf{sol}_3 \leq \mathsf{sol}_2 +1.$ Let
  \mindelay\ have a polynomial time approximation algorithm with
  solution $\mathsf{sol}_3 = y \cdot \mathsf{opt}_3.$ We know that
  $\mathsf{sol}_2 > \frac{10}{9}\mathsf{opt}_2.$ Substituting the
  relation between $\mathsf{sol}_3$ and $\mathsf{sol}_2$ in the above
  expression we get, $ \mathsf{sol}_2 > \frac{10}{9}\mathsf{opt}_2$ and 
    $\mathsf{sol}_3 = y \mathsf{opt}_3 > \frac{10}{9}\mathsf{opt}_2.$ This implies that
    $y > \frac{10}{9} \frac{\mathsf{opt}_2}{\mathsf{opt}_3}.$
  Observe that, if $\mathsf{opt}_2 \geq 10$ then $\mathsf{opt}_3 \leq
  \mathsf{opt}_2 +1$ (because of the reduction) implies that
  $\mathsf{opt}_2 > \frac{1}{1.1}\mathsf{opt}_3.$ So we get $y >
  100/99.$ This completes the proof.
\end{proof}

We now consider the decision version of \mincost.

\begin{remark}
  \label{rm:directed}
  Recall that the cost of an embedding $\embedding$ is computed using
  \eqref{eq:embedding-cost} which does not consider the direction on
  the edges of $\compgraph.$ It only considers the weight $\edgewt$ on
  any edge of $\compgraph.$ Thus the cost of an embedding does not
  depend on the directions of the edges and the solution of \mincost\
  problem is same irrespective of whether the computation graph
  $\compgraph$ has directions or not.
\end{remark}

\begin{theorem}
  \label{thm:mincost}
  The decision version of the \mincost\ is NP-complete.
\end{theorem}

\begin{proof}
  We actually prove that the decision version of \mincost\ is NP-hard
  even when the processing costs are zero and the costs on the edges
  of $\compgraph$ and $\net$ are all one. In this case the cost of the
  embedding $\embedding$ is given by $C(\embedding) := \sum_{(\omega_i,\omega_j) \in \compedges} \distance_{\embedding(\omega_i)\embedding(\omega_j)},$
  where $\distance_{uv}$ is the shortest distance between the vertices
  $u,v \in \netnodes.$

  We prove the hardness of the decision version of \mincost\ by giving
  a reduction from the unweighted version of \textit{Multiterminal
    Cut} problem which is NP-complete \cite{Dahlhaus94}.
  Multiterminal Cut problem is defined as follows: Given an arbitrary
  graph $\net^{\prime} = (\netnodes^{\prime},\netedges^{\prime})$ and
  a set $S = \{s_1,s_2,\ldots,s_k\} \subset \netnodes^{\prime}$ of $k$
  specified vertices, find the minimum number of edges $\netedges_s
  \subset \netedges^{\prime}$ such that the removal of $\netedges_s$
  from $\netedges^{\prime}$ disconnects each vertex in $S$ from all
  the other vertices in $S.$

  The cost of an embedding can be computed in polynomial time using
  \eqref{eq:embedding-cost}. Hence, given an instance of the decision
  version of the \mincost\ problem, one can guess an embedding in a
  non-deterministic way and check whether its cost is less than $L$ or
  not in polynomial time. Thus the decision version of the \mincost\
  problem is in NP. To prove the NP-hardness of the problem we will
  first show a transformation of an instance of Multiterminal Cut
  problem $\psi = (\net^{\prime},S,D)$ to an instance of the decision
  version of \mincost\ $\phi=(\net,\compgraph,\gamma,\lambda,D).$ Then
  we will show that there exists a set of edges in $\net^{\prime}$ of
  size $D$ which separates all the vertices of $S$ from all other
  vertices in $S$ if and only if there is an embedding $\embedding$
  such that a vertex $\gamma_i \in \gamma$ is mapped to a vertex
  $\lambda_i \in \lambda$ and has cost equal to $D.$

  From $\psi$ define an instance of decision version of \mincost\
  $\phi$ as follows: Let $\net$ be a complete graph on $\netnodes=
  \{u_1,\ldots,u_k\}$ where $k=|S|,$ and $\compgraph = \net^{\prime}.$
  Note that in this case $\compgraph$ is undirected. Define
  $\gamma=S$ and $\lambda = \netnodes$ such that a vertex $\gamma_i$
  is mapped to $u_i \in \netnodes.$ In other words, $k$ distinct
  vertices of $\compnodes$ are mapped to distinct vertices of
  $\netnodes.$ Now we prove that there is an embedding $\embedding$ of
  cost $D$ if and only if $\psi$ has a Multiterminal Cut of size equal
  to $D.$

  The ``if'' part is easy to see. If there is a minimum Multiterminal
  Cut $\netedges_s$ of size $D$ which divides the vertex set
  $\netnodes^{\prime} = \netnodes_1^{\prime} \cup \ldots
  \netnodes_k^{\prime}$ into $k$ disjoint subsets then define
  $\embedding$ such that each vertex in $\netnodes_i^{\prime}$ is
  mapped to $u_i \in \netnodes.$ Then the cost of the embedding is the
  total number of edges which go from $\netnodes_i^{\prime}$ to
  $\netnodes_j^{\prime}$ for $i \neq j.$ This is nothing but the size
  of the set $\netedges_s$ which is equal to $D.$

  To complete the proof we need to show that if there is no minimum
  Multiterminal Cut of $\psi$ of size $D$ then there is no embedding
  (which maps the vertices of $\gamma$ to vertices of $\lambda)$ of
  $\phi$ with cost $D.$ Let us assume that there is no minimum
  Multiterminal Cut of $\psi$ of size $D$ but there is an embedding
  for $\phi$ with cost $D.$ It implies that there is a mapping of
  $\compnodes$ on $k$ different vertices of $\netnodes$ such that
  $\gamma_i=s_i$ is mapped to $u_i.$ Let us denote all the vertices of
  $\compnodes$ that are mapped to $u_i$ by $\compnodes_i.$ The cost of
  the embedding is equal to the number of edges between sets
  $\compnodes_i$ and $\compnodes_j$ for $i \neq j$ which is equal to
  $D.$ Now it is easy to see that if we divide the vertices of
  $\net^{\prime}$ in $k$ disjoint subsets such that all the vertices
  in $\compnodes_i$ are in the same set then we can create a
  Multiterminal Cut of the graph $\net^{\prime}$ which has cost
  exactly $D.$ But this contradicts our assumption. Hence if there is
  an embedding function $\embedding$ for $\phi$ of cost $D$ then there
  is a Multiterminal Cut of $\psi$ of same size. This proves that the
  decision version of \mincost\ is NP-hard when the computation graph
  is undirected. From Remark~\ref{rm:directed}, thus the decision
  version of \mincost\ with directed computation graph is also
  NP-hard.

  So far, we have not considered any weight functions. From
  \cite{Dahlhaus94} we know that the Multiterminal Cut problem for
  weighted graphs is also NP-complete. And a simple modification in
  our reduction will prove that decision version of \mincost\ problem
  with weight functions is also NP-hard.
\end{proof}

It is also shown in \cite{Dahlhaus94} that Multiterminal Cut problem
is MAX SNP-hard. To prove any problem being MAX SNP-hard it is
sufficient to give a linear reduction to it from a known MAX SNP-hard
problem. The linear reduction is defined as follows:

\begin{definition}
  Let $\Pi$ and $\Pi^{\prime}$ be two optimization problems. Then we
  say that $\Pi$ linearly reduces to $\Pi^{\prime}$ if there are two
  polynomial time algorithms $A,B$ and two constants $\alpha,\beta >
  0$ such that
  \begin{enumerate}
  \item Given an instance $\pi$ of $\Pi$ with an optimal cost
    $\mathsf{opt}(\pi)$ an algorithm $A$ produces an instance
    $\pi^{\prime} = A(\pi)$ of $\Pi^{\prime}$ such that the cost of an
    optimal solution for $\pi^{\prime} (\mathsf{opt}(\pi^{\prime}))$
    is at most $\alpha\ \mathsf{opt}(\pi),$ i.e., $\mathsf{opt}(\pi^{\prime}) \leq \alpha\ \mathsf{opt}(\pi).$
  \item Given $\pi, \pi^{\prime} = A(\pi)$ and any solution $y$ of
    $\pi^{\prime}$ there is an algorithm $B$ which produces a solution
    $x$ of $\pi$ such that $|\mathsf{cost}(x) - \mathsf{opt}(\pi)| \leq \beta |\mathsf{cost}(y) - \mathsf{opt}(\pi^{\prime})|.$
  \end{enumerate}
\end{definition}

Note that in the proof of Theorem~\ref{thm:mincost} we use polynomial
time algorithms to reduce an instance $\psi$ of the Multiterminal Cut
problem to an instance $\phi$ of the \mincost\ problem and we proved
that $\mathsf{cost}(\psi) = \mathsf{cost}(\phi).$ We can also get a
solution of $\phi$ from a solution of $\psi$ and vice versa with
parameters $\alpha = \beta = 1.$ Hence the reduction we used to prove
that \mincost\ is NP-complete is in fact a linear reduction of
Multiterminal Cut problem to \mincost. We thus have the following
corollary.

\begin{corollary}
  Problem \mincost\ is MAX SNP-hard and hence does not have any
  polynomial time approximation scheme unless P=NP \cite{Arora92}.
\end{corollary}

The NP-hardness of \mincost\ can also be proved by reduction from
another well known NP-complete problem $k$-clique \cite{Garey79}.  The
$k$-clique problem is defined as follows: Given an arbitrary graph
$\net^{\prime}$ and a positive integer $k$ check whether
$\net^{\prime} = (\netnodes^{\prime},\netedges^{\prime})$ has a clique
(or a complete subgraph) of size $k.$ Since we know that the
$k$-clique problem is $W[1]-$complete \cite{Downey99}, the reduction
from it also implies that \mincost\ does not have any fixed-parameter
tractable algorithm and is also hard for $W[1].$

A variation of \mincost\ is considered in \cite{Fernandez-Baca89} in
that the communication costs (computed as $\edgewt(\omega_i,\omega_j)
\cdot \distance_{\embedding(\omega_i)\embedding(\omega_j)}$ in this
paper) are not necessarily zero if $\embedding(\omega_i) =
\embedding(\omega_j).$ Further, the source and the sink nodes are not
fixed like in this paper. The key complexity result is Lemma~2.1 which
shows that their variation of \mincost\ with zero-one communication
costs is NP-complete by reducing it from the planar $3$-SAT
problem. Their proof technique does not allow the communication to be
necessarily zero if $\embedding(\omega_i) = \embedding(\omega_j).$ For
non zero-one communication costs, \cite{Fernandez-Baca89} also gives a
polynomial time algorithms for their variation when $\compgraph$ is a
partial $k$-tree and \textit{almost trees with parameter $k.$} 

As is the case with many NP-complete problems our problems also become
tractable when $\compgraph$ has special structures. We consider three
such structures that have wide applications---the tree, layered graphs
and bounded treewidth graphs. These are considered next.

\section{$\compgraph$ is a Tree}
\label{sec:tree}
Many functions that are useful on sensor networks, e.g., average,
maximum, minimum etc., can be represented by directed tree
graphs. Operations that are required to resolve many database queries
can also be represented as directed tree structures.  The trees
representing a $\compgraph$ is from a class of trees that have a set
of leaf vertices whose in-degree is zero and a root vertex whose
out-degree is zero. We consider a tree structured $\compgraph$ such
that all the leaf vertices represent the sources of data and the root
acts as the sink which wants to know the final function value. Recall
that we label all the vertices in $\compnodes$ as
$\{\omega_1,\ldots,\omega_p\}.$ Also $\prenodes{\omega_i}$ and
  $\sucnodes{\omega_i}$ represent the predecessor and successor
  vertices of the vertex $\omega_i \in \compnodes.$ It is easy to
verify that in this type of tree structured computation graph, every
vertex (except the root) has only one successor vertex, i.e., for any
$\omega \in \compnodes$ (except the root) the set $\sucnodes{\omega}$
is a singleton set. The set $\sucnodes{\omega_p}$ is null, where
$\omega_p$ is the root and there is a unique path from each source to
the sink.

As we have mentioned earlier, \cite{Ying08} and \cite{Shah13} adapt,
respectively, the Bellman-Ford and the Dijkstra shortest path
algorithms to solve \mincost\ when $\compgraph$ is a tree. In the
following we will describe Algorithm~\ref{algo:tree} that solves
\mindelay\ when $\compgraph$ is a tree.

\textbf{Algorithm overview:} Algorithm~\ref{algo:tree} is a
centralized algorithm which assumes knowledge of the all-pair shortest
path delay matrix $\mathtt{D}$ of the communication network $\net.$
The optimal embedding is computed by iterating through all the edges
of the computation graph $\compgraph.$ For an edge
$(\omega_i,\omega_j)$ of $\compgraph,$ the algorithm computes the
optimal delay of the path leading to the vertex $\omega_j$ from
sources via $\omega_i$ for all possible mappings of the vertex
$\omega_j$ in the network. The delay till any vertex $\omega_j$ is the
maximum of the optimal delays of all the paths reaching to $\omega_j$
plus the processing delay at that vertex. Once the optimal delay for
the sink node is computed, the algorithm backtracks to find the
optimal mapping of all the other vertices.

\textbf{Algorithm Description:} In each iteration the
Algorithm~\ref{algo:tree} maintains the following data structures: 
\begin{algorithm}
  \caption{Optimal embedding algorithm to solve \mindelay\ for tree
    graphs}
  \label{algo:tree}
 
  \begin{algorithmic} [1]
    \REQUIRE{Network graph $\net =(\netnodes,\netedges)$, $|\netnodes| = n,$ $|\netedges|= m,$ Weight function $T:\netedges \mapsto \mathbb{R}^{+}$, Tree computation graph $\compgraph = (\compnodes, \compedges ),$ $|\compnodes|= p,$ $|\compedges|  =p-1,$ Weight function $\edgewt:\compedges \mapsto \mathbb{R}^{+},$ Cost function $\processingwt:\compnodes \times \netnodes \mapsto \mathbb{R}^{+}.$} 
  
    \ENSURE{Embedding $\embedding$ with minimum delay under $T,\edgewt,$ and $\processingwt.$} 
    \STATE $[[\mathtt{D} =\distance_{uv}]]$ \COMMENT{$n \times n$ distance matrix for $\net.$}
    \BlankLine
    \COMMENT{Initialization of tables}
    \STATE $h_l(v) := 0 $ $\ \ \forall v \in \netnodes, l \in [1,p-1]$ \;
    \STATE $f_l(u,v) := 0 $ $\ \ \forall u,v \in \netnodes, l \in [1,p-1]$ \; 
    \FOR{$l=1$ \KwTo $K$}  
    \FORALL{$v \in \netnodes$}
    \STATE $f_l(s_l,v) := \edgewt(\omega_{l},\sucnodes{\omega_l})\distance_{uv}$ \;
    \STATE $h_l(v) \leftarrow f_l(s_l,v)$ \;
    \STATE $x_l(v) \leftarrow s_l$ \;
    \ENDFOR
    \ENDFOR
    \FOR{$l=K+1$ \KwTo $(p-1)$} 
    \FORALL{$v \in \netnodes$}
    \FORALL{$u \in \netnodes$}
    \STATE $f_l(u,v) :=  \processingwt(\omega_{l},u) + \edgewt(\omega_{l},\sucnodes{\omega_l})\distance_{uv}$ \;
    \ENDFOR
    \STATE $h_l(v) \leftarrow \min\limits_{u} \left\{\max[h_i(u)|\omega_i \in \prenodes{\omega_l}] + f_l(u,v)\right\}$ \;
    \STATE $x_l(v) \leftarrow \argmin\limits_{u} \left\{\max[h_i(u)|\omega_i \in \prenodes {\omega_l}] + f_l(u,v)\right\}$ \;
    \ENDFOR
    \ENDFOR
    \STATE $d(\embedding):= \max[h_{i}(t)|\omega_i \in \prenodes{\omega_p}] + \processingwt(\omega_p,t) $ \;
    \STATE $\embedding(\omega_p) = t $ \;
    \COMMENT{Backtracking}
    \FOR{$l=(p-1)$ \KwTo $1$}
    \STATE $\embedding(\omega_{l}) = x_l(\embedding(\sucnodes{\omega_l})) $ \;
    \ENDFOR
  \end{algorithmic}
\end{algorithm}

\begin{enumerate}
\item $f_l(u,v):$ It is the delay associated with edge
  $(\omega_{l},\sucnodes{\omega_l})$ and vertex $\omega_{l}$ when
  $\omega_{l}$ and $\sucnodes{\omega_l}$ are mapped to vertex $u$ and
  $v,$ respectively.
\item $h_l(v):$ It is the optimal delay of the path leading to vertex
  $\sucnodes{\omega_l}$ (via $\omega_l)$ when it is mapped to $v \in
  \net.$ The algorithm also stores the mapping of vertex $\omega_{l}$
  in $x_l(v)$ corresponding to this value.
\end{enumerate}

After initializing these data structures to zero (lines $2,3$) the
algorithm completes in the following two steps.

%
\textbf{Lines $4$--$10:$} This is the iteration over all the sources in
  $\compgraph.$ As the mapping of source $\omega_i \in \compgraph$ is
  fixed to source $s_i \in \net,$ here we just calculate the minimum
  delay required to reach to all the vertices from the source $s_i.$
  Note that the processing delay at the source is zero, i.e.,
  $\processingwt(\omega_i,s_i) =0.$
  
\textbf{Lines $11$--$19:$} This is the main loop of the algorithm which
  runs over all the remaining vertices of $\compgraph$ starting from
  $\omega_{K+1}$ to $\omega_{p-1}.$ In each iteration $f_l(u,v)$ is
  updated for all possible mappings of vertices
  $\omega_{l},\sucnodes{\omega_{l}}$ (lines $13$--$15)$. The function
  $f_l(u,v)$ is computed by adding the following delay terms.
  \begin{itemize}
  \item $\processingwt(\omega_{l},u):$ This is the processing delay
    associated with the vertex $\omega_{l}$ when it is performed at
    vertex $u \in\netnodes.$
  \item $\edgewt(\omega_{l},\sucnodes{\omega_{l}})\distance_{uv}:$
    This is the communication delay associated with the edge
    $(\omega_{l},\sucnodes{\omega_{l}})$ when mapped to $u$ and $v$
    respectively.
 \end{itemize}
 Then $h_l(v)$ is updated in line $16.$ Note that
 $\max[h_i(u)|\omega_i \in \prenodes{\omega_l}]$ is the minimum delay
 till the vertex $\omega_{l}$ when it is mapped to $u.$ This is
 equivalent to the first term in right hand side of
 ~\eqref{eq:nodedelay} (Section~\ref{sec:prelims}, Page 3). This along
 with the processing delay $\processingwt(\omega_{l},u)$ gives the
 delay of the vertex $\omega_{l}$ when mapped to $u.$ As mentioned
 earlier the algorithm also stores the mapping $u$ in $x_l(v)$ which
 minimizes the $h_l(v).$
 
\textbf{Lines $20$--$24:$} Once all the delays are computed the algorithm
  computes the delay at the sink vertex $t$ (as the mapping of
  $\omega_p$ is fixed to $t$ in embedding $\embedding)$ and finds the
  mapping of vertices of $\compgraph$ on $\net$ which gives this value
  by backtracking from the sink to the sources.

\subsubsection{Analysis of Algorithm~\ref{algo:tree}}

\begin{theorem}
  \label{thm:mindelay_tree}
  Algorithm~\ref{algo:tree} solves \mindelay\ when $\compgraph$ is a
  tree and runs in $O(pn^2)$ time. Recall that $p$ and $n$ are the
  number of vertices in $\compgraph$ and $\net$ respectively.
\end{theorem}

\begin{proof}
  We give the proof of correctness of Algorithm~\ref{algo:tree} only
  when the computation graph is unweighed. The proof can easily be
  extended to the case when there are weights on the edges of
  $\compgraph.$ Recall that the delay of an embedding is defined
  recursively over all the vertices of $\compgraph$ starting from the
  sink vertex. It is sufficient to prove that at any iteration $l$ the
  algorithm computes the optimal delay of embedding the path from any
  source $\omega_i$ to an intermediate vertex $\sucnodes{\omega_l}$
  via $\omega_l$ for all possible embeddings of $\sucnodes{\omega_l}.$
  And then at the end it chooses the fixed mapping of the sink
  $\omega_p$ and traces back the optimal paths from sink to all the
  sources via the intermediate vertices. We will prove this
  inductively.

  Let $\tilde{x}_i$ be the assignment of $\omega_i$ in $\net$ and at
  any iteration $\sucnodes{\omega_l} = \omega_j$ for some $j \in
  [1,p].$ The optimal path from any source $\omega_i$ for $i \in
  [1,K]$ to its successor vertex $\sucnodes{\omega_i} = \omega_j$ is
  just the shortest path distance between $s_i$ and the vertex to
  which $\sucnodes{\omega_i}$ will eventually be mapped. This is equal
  to $\distance_{s_i \tilde{x}_j}.$ It is easy to verify that in the
  algorithm (line $11)$ this value is stored in $h_l(v) \ \ \forall l
  \in [1,K]$ data structure for all $v \in \netnodes.$ Assuming that
  the optimal delays till the $(l-1)^{th}$ run are calculated by the
  algorithm and stored in $h_i$s we will show that at $l^{th}$ run the
  algorithm computes the optimal delay. The optimal delay of the path
  from sources to $\sucnodes{\omega_l}$ via $\omega_l$ is given by
  $g_l(\tilde{x}_j) := \min_{\tilde{x}_l} \left\{d(\tilde{x}_l) +
      \distance_{\tilde{x}_l \tilde{x}_j} \right\},$ 
  where, $d(\tilde{x}_l)$ is the optimal delay of the path till
  $\omega_l.$ The optimal delay $d(\tilde{x}_l)$ can further be
  expanded and written in terms of the delay of its predecessors as:
  $d(\tilde{x}_l) := \min_{\tilde{x}_l} \left\{\max\limits_{\omega_i
        \in \prenodes{\omega_l}}[ d(\tilde{x}_i) + \distance_{\tilde{x}_i
        \tilde{x}_l} ]+ \processingwt(\omega_l,\tilde{x}_l) 
    \right\}.$
 Substituting value of $d(\tilde{x}_l)$ in $g_l(\tilde{x}_j)$ we get
  \begin{align}
    g_l(\tilde{x}_j) = \min_{\tilde{x}_l} \left\{
      \max\limits_{\omega_i \in \prenodes{\omega_l}}g_i(\tilde{x}_l) +
      \processingwt(\omega_l,\tilde{x}_l) + \distance_{\tilde{x}_l
        \tilde{x}_j} \right\}. \label{eq:gfinal}
  \end{align}
  Recall the line $14$ of Algorithm~\ref{algo:tree} which computes
  $f_l(u,v) = \processingwt(\omega_l,u) + \distance_{uv}.$ This is
  nothing but the last two terms of the right hand side of
  \eqref{eq:gfinal} when $\tilde{x}_l = u$ and $\tilde{x}_j = v.$ Now
  observe the line $16$ of Algorithm~\ref{algo:tree} which computes
  $h_l(v) = \min\limits_{u} \left\{\max[h_i(u)|\omega_i \in
    \prenodes{\omega_l}] + f_l(u,v)\right\},$ where $h_i(u)$ is the
  optimal delay of the path leading to $\omega_l$ via its predecessor
  $\omega_i.$ The optimal delay of the path $h_i(u) =
  g_i(\tilde{x}_l)$ for $\tilde{x}_l=u \in \netnodes.$ Hence
  Algorithm~\ref{algo:tree} indeed computes the optimal delay of the
  path leading to $\sucnodes{\omega_l}$ via $\omega_l$ for all
  possible mappings $v \in \netnodes$ of $\sucnodes{\omega_l}$ and
  stores it in $h_l(v)$ at iteration $l.$

  Recall that $\compgraph$ is a tree thus the total number of edges in
  it are $p-1$ and Algorithm~\ref{algo:tree} is executed once for each
  edge. In each iteration it computes the delay of an edge for all
  possible mappings of its end points which requires $n^2$ (line $14)$
  time where $n$ is the number of vertices in $\net.$ Then it adds the
  delay to the delay of its predecessors and chooses the one with
  minimum value ($h_{l}(v)$ in line $16)$ which requires $O(n)$
  time. Hence the time to complete one iteration is $O(n^2+n) =
  O(n^2).$ The total time complexity of Algorithm~\ref{algo:tree} is
  $O(pn^2).$
\end{proof}

\section{$\compgraph$ is a Layered Graph}
\label{sec:layered-graph}

We now consider the case when $\compgraph$ is a layered graph.  An
example of a layered computation graph is shown in
Fig.~\ref{fig:layered}. We assume that there are $r$ layers and number
of vertices in each layer is at most $k.$ The vertices at layer $l$
are labelled $\omega_{l1},\omega_{l2},\ldots,\omega_{lk}.$ The
directed graph has edge $(\omega_{ai},\omega_{bj})$ only if $b=a
\mbox{ or } a+1.$ Here we also assume that all the sources are at
layer one and there is only one sink $t$ on the last layer.
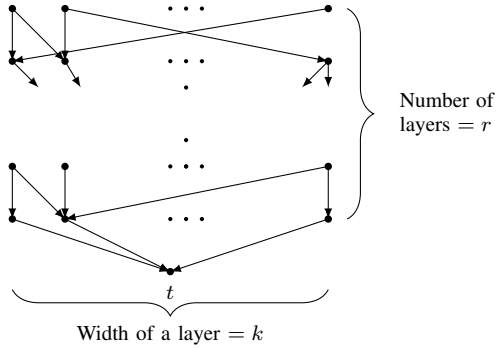
\begin{figure}[tbp]
  \centering
\begin{tikzpicture}[>=latex,scale=0.7]
\footnotesize
\foreach \x in {-3,-2,3} {
  \foreach \y in {-1,0,2,3} {
    \fill (\x,\y) circle (0.07cm);
  }  
}    
\foreach \x in {0,0.3,0.6} {
  \foreach \y in {-1,0,2,3} {
    \fill (\x,\y) circle (0.04cm);
  }  
} 
\foreach \y in {0,0.5,1.5}
  \fill (0.3,\y) circle (0.04cm);
\fill (0,-2) circle (0.07cm);

\draw[->] (-3,3) -- (-3,2);
\draw[->] (-3,3) -- (-2,2);  
\draw[->] (-2,3) -- (-2,2);     
\draw[->] (3,3) -- (-3,2);
\draw[->] (-2,3) -- (3,2);
\draw[->] (-3,2) -- (-2.5,1.5);
\draw[->] (-2,2) -- (-1.7,1.5);
\draw[->] (3,2) -- (2.5,1.5);
\draw[->] (3,2) -- (3,1.5);
\draw[->] (-3,0) -- (-3,-1);
\draw[->] (-3,0) -- (-2,-1);  
\draw[->] (-2,0) -- (-2,-1);     
\draw[->] (3,0) -- (-2,-1);
\draw[->] (3,0) -- (3,-1);
\draw[->] (-3,-1) -- (0,-2);
\draw[->] (-2,-1) -- (0,-2);
\draw[->] (3,-1) -- (0,-2);
\node at (0,-2.4){$t$};    
\draw [decorate,decoration={brace,amplitude=10pt},xshift=-4pt,yshift=0pt]
(3.5,3) -- (3.5,-1) node [black,midway,xshift=0.6cm,text width=2cm,right] {Number of layers $= r$};
\draw [decorate,decoration={brace,amplitude=10pt},xshift=0pt,yshift=-4pt]
(3,-2.2) -- (-3,-2.2) node [black,midway,yshift=-0.6cm] {Width of a layer $= k$};
\end{tikzpicture}
  \caption{Layered computation graph}
  \label{fig:layered}
\end{figure}  

We derive our motivation for these kind of computation graphs from the
MapReduce application framework. In MapReduce framework each user
comes to a network of processors with a set of Map and Reduce
tasks. There is a precedence order between Map and Reduce tasks. Each
Reduce task cannot be started unless the processing of corresponding
set of Map tasks is finished. Each task takes predefined time to
finish and the outputs of the Map tasks are used by the corresponding
Reduce tasks. This dependency can be represented by a directed graph
with edges showing the dependency between the two tasks. The aim in
this setting is to embed the Map and Reduce tasks on the processors
such that the total time of computation and communication is
minimized. We explain our motivation with the following example from
\cite{Dean04}.

\begin{example}
  \label{ex:mapreduce}
  Consider a typical database query by a server (call it sink) to
  check the number of occurrences of different words in two large
  files which are available at two separate servers. The
  task of calculating the number of occurrences of each word inside
  the files can be divided into the following sub-tasks.
  \begin{itemize}
  \item \textit{Splitting:} First the files are split into smaller
    sub-files such that each sub-file can be processed by a processor
    in the network.
  \item \textit{Mapping:} Each sub-file is then parsed to get the
    number of times each word occurred in it.
  \item \textit{Shuffling and reducing:} Once the count from each
    sub-file is available then the counts of one word are transported
    to one processor to compute the final count of that word. Each
    processor adds all the individual counts and results the final
    count of each word.
   \item \textit{Final result:} Finally the result is transported to
    the node which asked this query.
  \end{itemize}
  Fig.~\ref{fig:mapreduce} represents a typical MapReduce data flow
  diagram for this problem. The aim is to determine the processors for
  each of the sub-tasks such that the time to answer the query at the
  sink is minimized. The whole process can be represented by a
  directed layered graph with each layer representing one sub-task and
  a vertex at a layer representing a particular sub-task. Observe that
  the edges in this graph are only between the consecutive layers and
  the operations at a vertex cannot start until the data from all its
  predecessors is available.
  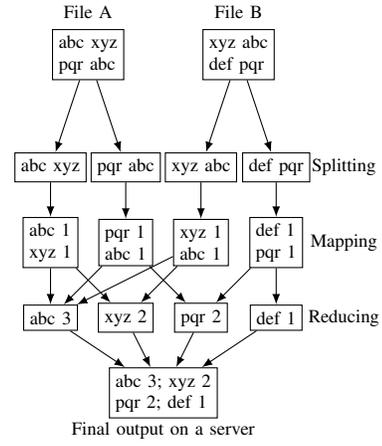
\begin{figure}[tbp]
    \centering
\begin{tikzpicture}[>=latex]
  \scriptsize
  \tikzstyle{ibox}=[draw=black,very thick,shape=rectangle,rounded corners=0.5em,inner sep=4pt,minimum height=2em,text badly centered];
  \node[draw,align=left](a) at (-1,2) {abc xyz \\ pqr abc};
  \node[draw,align=left](b) at (1,2) {xyz abc \\ def pqr};
  \node at (-1,2.4) [draw=none,above] {File A};
  \node at (1,2.4) [draw=none,above] {File B};  
  \node[draw,align=left](a1) at (-1.5,0.5) {abc xyz};
  \node[draw,align=left](a2) at (-0.5,0.5) {pqr abc};
  \node[draw,align=left](b1) at (0.5,0.5) {xyz abc};
  \node[draw,align=left](b2) at (1.5,0.5) {def pqr};
  \node at (2.4,0.5) [draw=none] {Splitting};
  \draw[->] (a) -- (a1);
  \draw[->] (a) --(a2);
  \draw[->] (b) -- (b1);  
  \draw[->] (b) -- (b2);
  \node[draw,align=left](a11) at (-1.5,-0.5) {abc 1 \\ xyz 1};
  \node[draw,align=left](a21) at (-0.5,-0.5) {pqr 1\\ abc 1};
  \node[draw,align=left](b11) at (0.5,-0.5) {xyz 1\\ abc 1};
  \node[draw,align=left](b21) at (1.5,-0.5) {def 1\\ pqr 1};
  \node at (2.4,-0.5) [draw=none] {Mapping};
  \draw[->] (a1) -- (a11);
  \draw[->] (a2) --(a21);
  \draw[->] (b1) -- (b11);  
  \draw[->] (b2) -- (b21);
  \node[draw,align=left](a12) at (-1.5,-1.5) {abc 3};
  \node[draw,align=left](a22) at (-0.5,-1.5) {xyz 2};
  \node[draw,align=left](b12) at (0.5,-1.5) {pqr 2};
  \node[draw,align=left](b22) at (1.5,-1.5) {def 1};
  \node at (2.4,-1.5) [draw=none] {Reducing};
  \draw[->] (a11) -- (a12);
  \draw[->] (a11)--(a22);
  \draw[->] (a21) -- (a12);
  \draw[->] (a21) -- (b12);
  \draw[->] (b11) -- (a12);
  \draw[->] (b11) -- (a22);
  \draw[->] (b21) -- (b12);
  \draw[->] (b21) -- (b22);
  \node[draw,align=left] (t) at (0,-2.5) {abc 3; xyz 2\\ pqr 2; def 1};
  \node at (0,-3) [draw=none,align=center] {Final output on a server};
   \draw[->] (a12) -- (t);
  \draw[->] (a22) --(t);
  \draw[->] (b12) -- (t); 
  \draw[->] (b22) -- (t);
\end{tikzpicture}
   \caption{Sub-tasks and data flow diagram of a typical database query
      in MapReduce framework}
    \label{fig:mapreduce}
  \end{figure}  
\end{example}

\begin{algorithm}\caption{Optimal embedding algorithm for layered graphs}
\label{algo:layer}
 
  \begin{algorithmic} [1]
    \REQUIRE{Network graph $\net =(\netnodes,\netedges)$, $|\netnodes| = n,$ $|\netedges|= m,$ Weight function $T:\netedges \mapsto \mathbb{R}^{+}$, Layered computation graph $\compgraph = (\compnodes, \compedges ),$ $|\compnodes|= p,$ $|\compedges|  =q,$ Weight function $\edgewt:\compedges \mapsto \mathbb{R}^{+},$ Cost function $\processingwt:\compnodes \times \netnodes \mapsto \mathbb{R}^{+}.$}  
  
    \ENSURE{Embedding $\embedding$ with minimum cost} 
    \STATE $[[\mathtt{D} =\distance_{u_i,u_j}]]$ \COMMENT{$n \times n$ distance matrix for $\net.$}
    \STATE $X := \{X_i: X_i \subseteq \netnodes$ and $|X_i| = k \}$  \COMMENT{$|X| = n^k,$ $X_i =\{a_1,a_2,\ldots ,a_k\}$}
    \STATE $Y := \{Y_i: Y_i \subseteq \netnodes$ and $|Y_i| = k \}$  \COMMENT{$|Y| = n^{k},$ $Y_i =\{b_1,b_2,\ldots ,b_k\}$}
    \STATE $Z_i := [z_{i1},z_{i2},\ldots,z_{ik}]$ \COMMENT{$z_{ij} \in \netnodes,$ is the assignment of node $\omega_{ij} \in \compnodes$ in optimal embedding}
    \BlankLine
    \COMMENT{Initialization of tables}
    \STATE $h_l(X_i) := 0 $ $ \ \ \forall X_i \in X, l \in [1,r]$ \;
    \STATE $f_l(X_i,Y_j) := 0 $ $\ \ \forall X_i \in X, Y_j \in Y, l \in [1,r]$ \; 
   
  \BlankLine
  \FOR{$l=1$ \KwTo $(r-1)$}
    \FOR{$j=1$ \KwTo $|Y|$}
      \FOR{$i=1$ \KwTo $|X|$}
    
      \STATE $f_l(X_i,Y_j) :=  \sum\limits_{\substack{u=1 \\ a_u \in X_i}}^{|X_i|} \processingwt(\omega_{lu},a_u) +
                        \sum\limits_{\substack{
                                      a_u \in X_i, b_v \in Y_j \\
                                      (\omega_{lu},\omega_{(l+1)v}) \in \compedges}} \edgewt(\omega_{lu},\omega_{(l+1)v}) \distance_{a_ub_v} + 
                        \sum\limits_{\substack{
                                      a_u,b_v \in X_i \\
                                      (\omega_{lu},\omega_{lv}) \in \compedges}} \edgewt(\omega_{lu},\omega_{lv}) \distance_{a_ub_v} +
                        h_{l-1}(X_i)$ \;
      \ENDFOR
      \STATE $h_l(Y_j) \leftarrow \min\limits_{X_i \in X} f_l(X_i,Y_j)$ \;
      \STATE $x_l(Y_j) \leftarrow \argmin\limits_{X_i \in X} f_l(X_i,Y_j)$ \; 
    \ENDFOR
  \ENDFOR
   \FOR{$i=1$ \KwTo $|X|$} 
 \STATE $h_r(X_i) := \sum\limits_{a_u,b_v \in X_i, 
                                      (\omega_{ru},\omega_{rv}) \in \compedges} \distance_{a_ub_v} +
               h_{r-1}(X_i)$ \;
                                      
  \ENDFOR
 \STATE $C(\embedding):= \min\limits_{X_j \in X} h_r(X_j) $ \;
 \STATE $ Z_r = \{z_{r1} , \ldots, z_{rk}\} = \argmin\limits_{X_j \in X} h_r(X_j)$ \; 
 \COMMENT{Backtracking}
  \FOR{$l=r-1$ \KwTo $1$}
     \STATE $Z_l = x_l(Z_{l+1}) $ \;
  \ENDFOR
\end{algorithmic}
\end{algorithm}

Now we present an algorithm which solves the \mincost\ problem for
layered graphs in polynomial time.

\textbf{Algorithm overview:} Like Algorithm~\ref{algo:tree},
Algorithm~\ref{algo:layer} also has two phases---a forward path and a
backward tracking. The forward path is a dynamic program that iterates
over all the layers of the computation graph. In the first iteration
it computes the optimal mappings of all the vertices of layer $1$
corresponding to a possible mapping of the vertices of layer $2.$ If a
vertex is a source vertex then its mapping is always fixed to the
corresponding source vertex in $\net.$ Similarly, in iteration $l$ it
computes the optimal mapping of vertices of layer $l$ each
corresponding to a possible mapping of vertices of layer $l+1.$ It
also computes the optimal cost till layer $l+1$ for every possible
mapping of vertices of layer $l+1.$ Once it reaches the last layer the
algorithm chooses the mapping of the vertices of last layer which
minimizes the overall cost and backtracks to get the corresponding
mappings of all the previous layers. 

\textbf{Algorithm Description:} Algorithm~\ref{algo:layer} iterates
over the layers in $\compgraph$ and in each iteration it maintains the
following two data structures:
\begin{enumerate}
\item $f_l(X_i,Y_j),$ the cost of embedding all the vertices till
  layer $l$ when the vertices at layer $(l+1)$ are placed at $Y_j$ and
  vertices of layer $l$ are placed at $X_i,$ where $X_i,Y_j \subset
  \netnodes$ of size $k.$
\item $h_l(Y_j),$ the optimal cost of embedding all the vertices (and
  corresponding edges) till layer $l$ when the vertices at layer
  $(l+1)$ are mapped to an ordered subset $Y_j \subset \netnodes$ of
  size $k.$
\end{enumerate}

After initializing these data structures to $0$ (in line $5,6$) the
algorithm completes in the following three steps.

%
\textbf{Lines $7$--$15$:} This is the main loop of the algorithm which
  runs for the first layer (from sources) to last but one layer. At
  each layer $l$ the data structure $f_l(X_i,Y_j)$ is updated for all
  possible combinations of $k$ size subsets $X_i$ and $Y_j$ of
  $\netnodes$ (lines $7$--$11$). The following cost terms are added
  together to calculate $f_l(X_i,Y_j)$ along with the optimal cost
  till layer $l-1,$ $h_{l-1}(X_i).$
  \begin{itemize}
  \item $\sum\limits_{\substack{u=1 \\ a_u \in X_i}}^{|X_i|}
    \processingwt(\omega_{lu},a_u):$ Cost of putting computation node
    $\omega_{lu} \in \compnodes$ at $a_u \in \netnodes$ for each node
    at the current layer.
  \item $\sum\limits_{\substack{a_u \in X_i, b_v \in Y_j \\
        (\omega_{lu},\omega_{(l+1)v}) \in \compedges}}
    \edgewt(\omega_{lu},\omega_{(l+1)v}) \distance_{a_ub_v}:$ Total
    communication cost, when node $\omega_{lu} \in \compnodes$ is
    placed at node $a_u \in \netnodes$ and $\omega_{(l+1)v} \in
    \compnodes$ is placed at $b_v \in \netnodes,$ is the
    multiplication of corresponding costs in computation graph (weight
    $\edgewt$) and communication graph (weight $T$). This term
    captures the cost for all the edges between layer $l$ and layer
    $l+1.$
  \item $\sum\limits_{\substack{ a_u,b_v \in X_i \\
        (\omega_{lu},\omega_{lv}) \in \compedges}}
    \edgewt(\omega_{lu},\omega_{lv}) \distance_{a_ub_v}:$ Total
    communication cost, when node $\omega_{lu} \in \compnodes$ is
    placed at node $a_u \in \netnodes$ and $\omega_{lv} \in
    \compnodes$ is placed at $b_v \in \netnodes,$ is the
    multiplication of corresponding costs in computation graph (weight
    $\edgewt$) and communication graph (weight $T$). This term
    captures the cost for all the edges at layer $l.$
  \end{itemize}
  %
  
\textbf{Lines $16$--$20$:} Here the algorithm finally computes the total
  cost of the embedding the graph $\compgraph$ by adding the cost of
  the edges between the vertices of last layer $r,$ if any, when the
  vertices of last layer are placed at $X_i.$ And computes the optimal
  cost of embedding $\embedding,$ $C(\embedding),$ by choosing the
  placement of last layer which minimizes the overall cost (line
  $19$--$20$). The vector $Z_r$ stores the mapping of vertices at
  layer $r$ under the embedding $\embedding.$
  
\textbf{Lines $21$--$23$:} After finding the optimal mapping for the vertices
  at layer $r$ the algorithm traces back the corresponding optimal
  mapping for vertices at layer $r-1$ all the way upto the first
  layer.

To simplify the description of the algorithm we do not show the fixed
mapping of sources and sink into the network graph in
Algorithm~\ref{algo:layer}. It is easy to see that to incorporate the
fixed mapping of sources and sink the algorithm needs to only consider
those subsets $X_i,Y_j$ of $\netnodes$ which map the source $\omega_i$
(sink $\nu$) to the corresponding source $s_i$ (sink $t$) while
calculating the data structures of the layer on which $\omega_i$ is
present (layer $r$).

\subsection{Analysis of Algorithm \ref{algo:layer}}

\begin{theorem}
  \label{thm:mincost_layer}
  Algorithm~\ref{algo:layer} solves \mincost\ and the time complexity
  of the algorithm is $O(rn^{2k})$ when $\compgraph$ is a layered
  graph with $r$ layers and at most $k$ nodes per layer.
\end{theorem}

\begin{proof}
  We prove the correctness of \ref{algo:layer} when $\compgraph$ is an
  unweighed graph and the processing costs are all zero.  We also
  assume that there are no edges between the vertices of a layer. The
  proof can easily be extended with weight functions $\edgewt$ and $\processingwt.$ It
  is sufficient to show that at each iteration $l$ the algorithm
  computes the optimal cost of embedding the computation graph till
  layer $l$ for all possible embeddings of every node of layer $(l+1)$
  given that the cost computed till layer $(l-1)$ is optimal. Let
  $\tilde{x}_i$ be a vector of size $1 \times k$ whose $l^{th}$
  element represents the assignment of a network node for $\omega_{li}
  \in \compnodes.$ In other words, $\tilde{x}_i$ is a $k$ size subset
  of $\netnodes.$ Let us define,
  \begin{equation}
    d(\tilde{x}_i,\tilde{x}_{i+1}) := \sum\limits_{\substack{a_u \in
        \tilde{x}_i, b_v \in \tilde{x}_{i+1}
        \\ (\omega_{iu},\omega_{(i+1)v}) \in \compedges}}
    \distance_{a_ub_v}. \label{dist}
  \end{equation}
  This represents the sum of distance between all the adjacent nodes
  in layer $i$ and $(i+1)$ when the nodes of layer $i$ are embedded to
  $\tilde{x}_i$ and nodes of layer $(i+1)$ are embedded to
  $\tilde{x}_{i+1}.$ Total cost of any embedding can then be written
  as $C := \sum\limits_{ i=1}^{r-1} d(\tilde{x}_i,\tilde{x}_{i+1}). $
  To obtain the optimal embedding we have to minimize the above
  equation with respect to all the possible mappings $\tilde{x}_i.$
  Therefore the optimal cost can be written as $C_{opt} = \min_{\tilde{x}_1,\ldots,\tilde{x}_r} \left(\sum\limits_{i=1}^{r-1} d(\tilde{x}_i,\tilde{x}_{i+1})\right).$
  Separating the terms with $\tilde{x}_1$ and some algebraic
  manipulation will give us
  \begin{eqnarray}
    C_{opt} &=& \min_{\tilde{x}_2} g_1(\tilde{x}_2) + 
    \min_{\tilde{x}_2,\ldots,\tilde{x}_r} \left(\sum\limits_{
      i=2}^{r-1} d(\tilde{x}_i,\tilde{x}_{i+1})\right), \label{g_1}  
  \end{eqnarray}     
  where, $g_1(\tilde{x}_2) = \min_{\tilde{x}_1}
  d(\tilde{x}_1,\tilde{x}_2).$ Similarly after minimizing with respect 
  to $\tilde{x}_l$ we can write the cost as,
  \begin{eqnarray}
    C_{opt} &=&  \min_{\tilde{x}_{l+1},\ldots,\tilde{x}_r}
    \left(g_l(\tilde{x}_{l+1}) + \sum\limits_{ i=l+1}^{r-1}
    d(\tilde{x}_i,\tilde{x}_{i+1})\right), \label{g_l} 
  \end{eqnarray}
  where $g_l(\tilde{x}_{l+1}) = \min\limits_{\tilde{x}_l}
  g_{l-1}(\tilde{x}_l).$ Now it suffices to show that the algorithm
  indeed calculates $g_l$ at $l^{th}$ iteration. Recall that
  $\tilde{x}_l$ and $\tilde{x}_{l+1}$ are $k$ size subsets of
  $\netnodes$ which represent the mapping of nodes of layer $l$ and
  $l+1$ respectively. In the algorithm, $X_i$ represents a $k$ size
  subset of $\netnodes$ to which the nodes of the layer of the current
  iteration are mapped. In other words, $X_i$ is the same as
  $\tilde{x}_l$ of the above discussion. Similarly $Y_j$ is the same
  as $\tilde{x}_{l+1}.$ Note that in the first iteration the algorithm
  calculates $f_1$ and $h_1$ as follows:
  \begin{eqnarray}
    f_1(X_i,Y_j) &=& \sum\limits_{\substack{a_u \in X_i, b_v \in Y_j \\
        (\omega_{1u},\omega_{2v}) \in \compedges}} \distance_{a_ub_v} +
    h_0(X_i),  \label{f_1}  
  \end{eqnarray}
  for all $k$ size subsets $X_i$ and $Y_j$ of $\netnodes.$ As
  $h_0(X_i)$ is initialized to zero for all $X_i$ using \eqref{dist}
  we can write \eqref{f_1} as  $f_1(X_i,Y_j) = d(X_i,Y_j) \; \ \ \forall X_i \in X, Y_j \in Y.$
  Finally, $h_1$ is calculated by minimizing $f_1$ over $X_i.$
  \begin{eqnarray}
    h_1(Y_j) = \min\limits_{X_i \in X} f_1(X_i,Y_j) = \min\limits_{X_i
    \in X} d(X_i,Y_j). \label{h_1}
  \end{eqnarray} 
  By comparing \eqref{g_1} and \eqref{h_1} we get $h_1(Y_j) =
  g_1(\tilde{x}_2),$ when $Y_j =\tilde{x}_2.$ The algorithm maintains
  a table of $h_1$ and the value of $X_i$ for which $f_1(X_i,Y_j)$ is
  minimized for all the $k$ size subsets $Y_j.$ This table is
  equivalent to storing the value of $g_1(\tilde{x}_2)$ for all
  possible values of $\tilde{x}_2.$ Similarly at $l^{th}$ iteration
  the algorithm computes the following two terms:
  \begin{eqnarray*}
    f_l(X_i,Y_j) &=& d(X_i,Y_j) + h_{l-1}(X_i) \; \ \ \forall X_i \in X, Y_j \in Y \\
    h_l(Y_j) &=& \min\limits_{X_i \in X} f_l(X_i,Y_j) \; \ \ \forall Y_j \in Y.
  \end{eqnarray*}  
  The algorithm stores the table of $h_l$ and corresponding $X_i$ for
  each $Y_j.$ As $h_l(Y_j) = g_l(\tilde{x}_{l+1}),$
  the algorithm exactly calculates $g_l$ at each iteration and
  maintains a table for all possible embeddings for nodes at layer $l$
  for each embedding of nodes at layer $l+1.$ This is the same as
  minimizing with respect to one $\tilde{x}$ at a time as explained in
  \eqref{g_l}. The computation of tables for $h_l$ only depends on the
  local variables, i.e., it only depends on the edges between layer
  $l$ and $l+1$ and all possible embeddings of nodes of layer $l$ and
  layer $l+1.$

  As there are at most $k$ nodes at each layer and there are $n$
  possible locations where each node of computation graph can be
  placed in the communication graph, the time required to compute
  $f_l$ is $O(n^k \times n^k) = O(n^{2k})$ and the time to compute the
  corresponding $h_l$ is $O(n^k).$ Thus total time to compute the
  table $h_l$ for a layer is $O(n^{2k}+n^{k}).$ There are $r$ layers
  in the computation graph hence the computation of $h_l$ table is
  done at most $r$ times which gives the time complexity of the
  algorithm as $O(r (n^{2k}+n^{k}))=O(rn^{2k}).$
\end{proof}

We now claim that output of Algorithm~\ref{algo:layer} is a $k^2$
approximation of \mindelay. Let the cost obtained from the embedding
$\embedding_{opt}^c$ be $C_{opt}^{c}.$ Once an embedding is given, we
can obtain the delay of the embedding by recursively using
\eqref{eq:nodedelay} to find the delay at the sink. Let the delay of
the embedding $\embedding_{opt}^c$ be $D_{opt}^{c}.$ Note that in
finding the delay of any vertex in the embedding we take the maximum
of the delays coming from all its incoming edges, i.e., if the delays
of the incoming edges are $d_1,d_2,\ldots,d_k$ then the delay at the
vertex is $d = \max(d_1,\ldots,d_k).$ On the other hand while
computing the cost at any vertex we add the costs coming from all its
incoming edges, i.e., $c = d_1+\ldots+d_k.$ Hence at any vertex $d
\leq c.$ This implies that for any embedding $\embedding,$ $d(\embedding) \leq C(\embedding).$
Between any two layers of a bounded width computation graph there
are at most $k^2$ edges and if we assume that the delay on each edge
is same then the cost at any vertex is $c =k^2d.$ With the same logic
one can easily prove that for an embedding $\embedding,$ $C(\embedding) \leq k^2 d(\embedding).$
Thus for the minimum cost embedding $\embedding_{opt}^c$
\begin{equation}
  \frac{C_{opt}^{c}}{k^2} \leq D_{opt}^{c} \leq C_{opt}^{c}. \label{eq:costopt}
\end{equation}
Let $\embedding_{opt}^d$ be the embedding which minimizes the delay of
$\compgraph$ on $\net$ with $D_{opt}^{d}$ and $C_{opt}^{d}$ being its delay and cost respectively. Then we know that
\begin{equation}
  \frac{C_{opt}^{d}}{k^2} \leq D_{opt}^{d} \leq
  C_{opt}^{d}. \label{eq:delayopt}
\end{equation}
As $\embedding_{opt}^d$ minimizes the delay, $D_{opt}^{d} \leq D_{opt}^{c}.$ From \eqref{eq:costopt} $D_{opt}^{d} \leq D_{opt}^{c} \leq C_{opt}^{c}.$
Similarly $C_{opt}^{c} \leq C_{opt}^{d}$ which along with
\eqref{eq:delayopt} gives $\frac{C_{opt}^{c}}{k^2} \leq \frac{C_{opt}^{d}}{k^2} \leq D_{opt}^{d}.$
Finally we get,
\begin{equation}
  \frac{C_{opt}^{c}}{k^2} \leq D_{opt}^{d} \leq C_{opt}^{c}. \label{eq:delaycost}
\end{equation}
This implies that the cost of $\embedding_{opt}^c$ is a $k^2$
approximation of the delay of $\embedding_{opt}^d.$ We have thus shown
the following. 
\begin{theorem}
  Algorithm~\ref{algo:layer} gives an embedding whose delay is the
  $k^2$ approximation of \mindelay\ when $\compgraph$ is a layered
  graph with $r$ layers and has at most $k$ nodes per layer.
\end{theorem}

\section{$\compgraph$ is a Bounded Treewidth Graph}
\label{sec:treewidth}

We now extend the application of the algorithm of the preceding
section to a graph that may not be a DAG. Towards that, we use the
notion of the \textit{treewidth} of the graph, which is a measure of
how far the graph is from a tree. The following definition of the
treewidth of a graph is from \cite{Diestel00} and reproduced here for
the sake of completeness.

\begin{definition}
  A tree decomposition of a graph $\compgraph=(\compnodes,\compedges)$
  is a tree $T$ with vertices $V_1,\ldots,V_r$ such that each $V_i
  \subset \compnodes$ and satisfies the following properties:
  \begin{enumerate}
  \item $\cup_i V_i = \compnodes.$
  \item If $u \in V_i$ and $u \in V_j$ then $u \in V_k$ for all $V_k$
    such that $V_i,V_k,V_j$ form a connected component.
  \item For all $(u,v) \in \compedges$ there exists a subset $V_i$
    such that both $u,v \in V_i.$
  \end{enumerate}
  The width of a tree decomposition is the size of largest $V_i$ minus
  one. The treewidth $tw$ of a graph is the minimum width among all
  possible tree decomposition of the graph.
\end{definition}

In the previous section we presented Algorithm~\ref{algo:layer} to
find the minimum cost embedding of a layered graph when the edges are
possible only between the consecutive layers. It is easy to observe
that the treewidth of such a layered graph with maximum width $k$ is
$tw = 2k-1.$ The tree decomposition of the layered graph is shown in
Fig.~\ref{fig:layertree}. A simple reinterpretation of the process of
finding the minimum cost embedding in Algorithm~\ref{algo:layer} gives
us a procedure to find the minimum cost embedding of the graphs with
bounded (constant in terms of the size of the graph) treewidth.

Let us denote the vertices in the tree decomposition of the layered
graph as $V_1,\ldots,V_{r-1},$ where the vertex $V_i$ contains all the
vertices from layer $i$ and $(i+1).$
\begin{figure}[tbp]
  \centering
   \subfloat[]{
\begin{tikzpicture}[>=latex,scale=0.7]
\scriptsize
\draw (5,4) node[draw,rounded corners,align=left] (a) {$\omega_{11},\ldots,\omega_{1w}$\\$\omega_{21},\ldots,\omega_{2w}$};
\draw (5,2) node[draw,rounded corners,align=left] (b) {$\omega_{21},\ldots,\omega_{2w}$\\$\omega_{31},\ldots,\omega_{3w}$};
\draw (5,-1) node[draw,rounded corners,align=left] (c) {$\omega_{(r-1)1},\ldots,\omega_{(r-1)w}$\\$\omega_{r1},\ldots,\omega_{rw}$};
\draw [-] (a) --(b);
\draw [-] (b)  -- (5,1);
\draw [-] (c) -- (5,0);
\fill (5,1) circle (0.07cm);
\fill (5,0.5) circle (0.07cm);
\fill (5,0) circle (0.07cm);
\end{tikzpicture}
}
  \caption{Tree decomposition of the layered graph with $r$ layers of
    width $k$ each}
  \label{fig:layertree}
\end{figure}
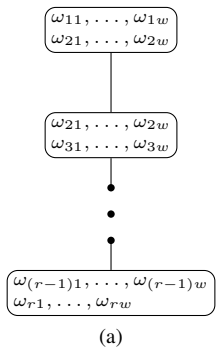
Observe that in the $i^{th}$ run of the loop written in lines
$9$--$11,$ Algorithm~\ref{algo:layer} computes the cost of embedding
all the edges and nodes present in the vertex $V_i$ of the tree
decomposition for all possible mappings of nodes in $V_i.$ In lines
$12$--$13$ the algorithm finds the optimal embedding cost of all the
nodes and edges till $V_i$ conditioned on the mapping of vertices in
$V_i \cap V_{i+1}.$ Lines $16$--$18$ compute the cost till the last
layer and then the algorithm traces back from $V_{r-1}$ to $V_1$ to
get the final optimal cost and the corresponding embedding (lines
$19$-$23$). Note that the time required to compute the cost till
vertex $V_i$ depends on the size of this vertex (lines $8,9$ define
that value). And the total time to complete the process can be written
as $O(rn^{tw+1})$ where $tw$ is the size of largest $V_i$ in a tree
decomposition.

Therefore if we can find the tree decomposition of any computation
graph $\compgraph$ of $r$ vertices with size of largest vertex to be
$tw$ then an algorithm similar to Algorithm~\ref{algo:layer} can be
used to compute the minimum cost embedding in time $O(rn^{tw+1}).$

\begin{figure}[tbp]
 \centering
\subfloat[]{
\begin{tikzpicture}[>=latex,scale=0.7]
 \scriptsize
 \tikzstyle{every node} = [circle,draw=black]
 \node (s) at (0,3) {$S$};
 \node (a) at (0,2) {$a$};
 \node (b) at (0,1) {$b$};
 \node (c) at (-1,0) {$c$};
 \node (d) at (1,0) {$d$};
 \node (e) at (0,-1) {$e$};
 \node (t) at (0,-2) {$t$};
 \draw [->] (s) -- (a);
 \draw [->] (a) --(b);
 \draw [->] (b) --(c);
 \draw [->] (b) --(d);
 \draw [->] (c) --(e);
 \draw [->] (d) --(e);
 \draw [->] (e) --(t);   
 \draw [->] (e) to[out=0,in=-30,looseness=2] (a);
\end{tikzpicture}
}
\hspace*{10pt}
\subfloat[]{
\begin{tikzpicture}[>=latex,scale=0.7]
  \scriptsize
  \tikzstyle{every node} = [circle,draw=black]
  \node (a) at (0,3) {$s,a$};
  \node (b) at (0,1) {$a,b,e$};
  \node (c) at (2,-1) {$b,d,e$};
  \node (d) at (1,-3) {$e,t$};
  \node (e) at (-2,-1) {$b,c,e$};
  \draw [-] (a) -- (b) -- (c) -- (d);
  \draw [-] (b) -- (e);
\end{tikzpicture}
} 
 \caption{(a). A simple computation graph with directed
    cycle. (b). Tree decomposition of the graph}
  \label{fig:series}
\end{figure}
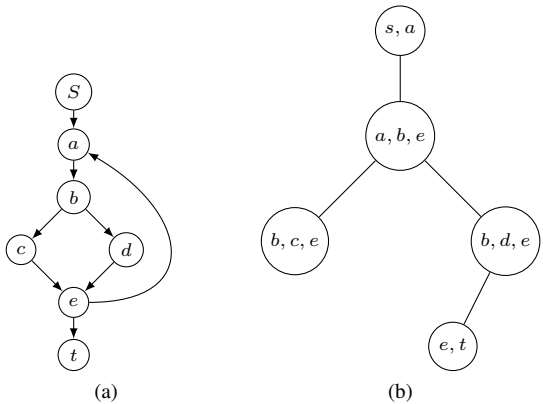

An example of a non DAG graph for the technique discussed above is
shown in Fig.~\ref{fig:series}a. It is a computation
schema that has a conditional jump, represented by the link between
vertexes $e$ and $a$. This is an example of a series-parallel graph
\cite{Diestel00} and such graphs are amenable to a tree decomposition
like in Fig.~\ref{fig:series}b.  In the preceding we have shown that
the technique used in Algorithm~\ref{algo:layer} can also be used to
find the minimum cost embedding of a series-parallel computation
graph, among others, that can have a bounded tree decomposition.

\section{Updating Solution to \mincost\ for Perturbations of $\compgraph$}
\label{sec:change_computation}

Let us consider a situation where the minimum cost embedding for a
layered graph $\compgraph$ is given and one needs to find the
embedding for a new graph $\compgraph^{\prime}$ which is generated by
adding vertices and/or edges in $\compgraph.$ We assume that
$\compgraph^{\prime}$ is still a layered graph with $r$ layers and
maximum width $k.$ Assume that we are given a set of $r$ tuples
$(e_i,l_i)$ where edge $e_i$ is added at layer $l_i.$ Note that edge
$e_i$ should have at least one end point at the existing vertex in
graph $\compgraph.$ To find the new embedding we first sort the $r$
tuples in $[(e_1,l_1),\ldots,(e_r,l_r)]$ such that $l_1 \leq l_2
\ldots \leq l_r.$ Now we start adding the edges layer wise from $l_1$
to $l_r.$ At any layer $l$ the following three types of additions are
possible.

\begin{enumerate}
\item \textbf{Addition of a vertex with only one edge.} When a vertex,
  say $v,$ is added with an edge $uv$ to an existing vertex $u$ at
  layer $l,$ then $v$ can be seen as a sink to an intermediate
  function value available at $u.$ Let us assume that the vertex $u$
  is mapped to vertex $z_u \in \netnodes$ under the original
  embedding. If the mapping $z_v \in \netnodes$ of $v$ is predefined
  (which is generally the case for all the sources and sinks in the
  network) then we just have to find the minimum cost path between
  $z_u$ and $z_v$ and add it to the existing embedding to find the new
  embedding. If the mapping of $v$ is not predefined then mapping it
  to $z_u$ will give the minimum cost embedding. This can be done in
  $O(1)$ time.

\item \textbf{Addition of an edge between two existing vertices.} Let
  us consider a situation when an edge $uv$ is added such that $u$ is
  at layer $l$ and $v$ is at layer $(l+1)$ with weight $\edgewt(u,v).$
  The data structure already available for each layer $l$ after
  running Algorithm~\ref{algo:layer} is shown in
  Fig.~\ref{fig:algo}. Recall that $f_l(X_i,Y_j)$ is the cost of
  embedding till layer $l$ (including the edges between layer $l$ and
  $(l+1)$) when vertices of layer $l$ are at $X_i$ and that of $(l+1)$
  are at $Y_j.$ And $h_l(Y_j)$ is the optimal cost till layer $l$ when
  vertices at layer $(l+1)$ are at $Y_j.$ If the vertex $u$ is placed
  at $x_u \in X_i$ and $v$ is placed at $y_v \in Y_j$ then the new
  cost of embedding is $f_l^{\prime}(X_i,Y_j) = f_l(X_i,Y_j) +
  \edgewt(u,v)\distance_{x_uy_v}$ (line number $10$ of
  Algorithm~\ref{algo:layer}). One needs to modify the whole $f_l$
  matrix at layer $l$ by adding a value and correspondingly the
  minimum cost $h_l(Y_j)$ will also change (line number $12,13$). As
  the value of $h_l(X_i)$ changes the pointers to compute $f,$ the
  values for all the subsequent layers will also change. Hence at each
  layer $i$ we will modify the whole data structure just by adding new
  values of $h_{i-1}$ and subtracting the old values of $h_{i-1}.$
  Once all the data structures are changed the algorithm needs to run
  the same back tracking procedure (lines $19$--$23)$ to get the new
  embedding. Assuming that the modification in the data structure at
  each layer can be done in $O(T)$ time then the new embedding can be
  found in $O(rT)$ time.

\item \textbf{Addition of a vertex with more than one edge.} Let a
  vertex $v$ is added to layer $l$ with more than one edges to
  existing vertices at layer $(l-1)$ and/or $(l+1).$ The width of the
  layer $l$ is still upper bounded by $k.$ By following the same logic
  presented above it is easy to observe that now the data structures
  from layer $(l-1)$ onward will change, i.e., $f_{l-1}(X_i,Y_j)$ and
  $h_{l-1}(Y_j)$ will also change. And the new embedding can again be
  found in $O(rT)$ time.
\end{enumerate}

Note that at a layer $l$ the data structure changes due to edges added
to layers $1,\ldots,(l-1)$ (which takes $O(l-1)$ time) and the edge
added at layer $l$ (which takes only $O(T)$ time as described
earlier). Hence the total time to change all the data structures in
this process will be just $O(rT),$ where $r$ is the total number of
layers, as opposed to $O(r^2T)$ if we add each edge separately.
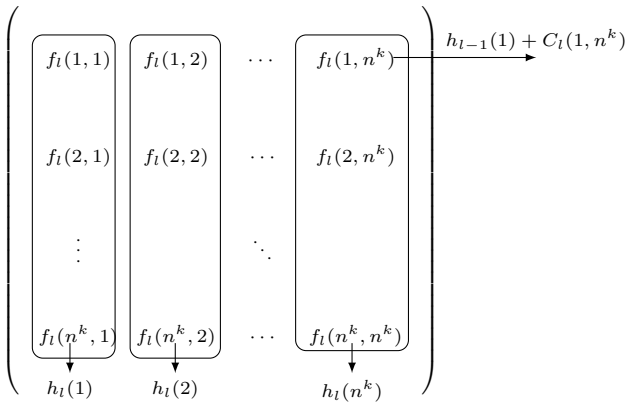
\begin{figure}[tbp]
  \centering
\begin{tikzpicture}[>=latex]
\scriptsize
 \matrix (A) [matrix of math nodes,%
             nodes = {node style ge},%
             left delimiter  = (,%
             right delimiter = )] at (0,0)
             {
 f_l(1,1) & f_l(1,2) & \cdots & f_l(1,n^k) \\
 f_l(2,1) & f_l(2,2) & \cdots & f_l(2,n^k) \\
 \vdots & & \ddots \\
 f_l(n^k,1) & f_l(n^k,2) & \cdots & f_l(n^k,n^k) \\
};
 \draw [->] (2.3,1.9) -- (4.2,1.9);
 \node at (4.2,1.9) [draw=none,above] {$h_{l-1}(1) +C_l(1,n^k)$};

 \draw [rounded corners] (1,-2) rectangle (2.5,2.2);
 \draw [->] (1.75,-1.9) -- (1.75,-2.3);
 \node at (1.75,-2.3) [draw=none,below] {$h_l(n^k)$};
 
 \draw [rounded corners] (-2.5,-2.1) rectangle (-1.4,2.2);
 \draw [->] (-2,-1.9) -- (-2,-2.3);
 \node at (-2,-2.3) [draw=none,below] {$h_l(1)$};  
 
  \draw [rounded corners] (-1.2,-2.1) rectangle (0,2.2);
  \draw [->] (-0.6,-1.9) -- (-0.6,-2.3);
 \node at (-0.6,-2.3) [draw=none,below] {$h_l(2)$};
\end{tikzpicture}
 \caption{Data structure in Algorithm~\ref{algo:layer} at layer $l$}
  \label{fig:algo}
\end{figure}

\section{Discussion}
\label{sec:discuss}

\subsection{Distributed versions of the algorithms.}  

Note that Algorithms~\ref{algo:tree} and \ref{algo:layer} are both
centralized algorithms. In other words they both need knowledge of
$\net.$ The algorithms have two stages with the all-pair distance
matrix $\mathtt{D}$ of $\net$ computed in the first stage and the
optimal value of $h_l$ and $f_l$ using $\mathtt{D}$ computed in the
second stage; see lines $6$ and $14$ of algorithm~\ref{algo:tree} and
lines $10$ and $17$ of Algorithm~\ref{algo:layer}.  Observe that once
$\mathtt{D}$ is known, the outputs of our algorithms are independent
of the vertex of the communication network which runs the
algorithm. Hence all vertices in $\net$ can run the algorithms to
obtain the optimal embedding without interacting with the other
vertices in the network. Hence if one can compute the distance matrix
in a distributed way in the network, then the optimal embedding can
also be found in a distributed manner. Several distributed algorithms
to find $\mathtt{D}$ are available, e.g., \cite{Kanchi04}.

\subsection{Delay in the network with bounded capacity.}

In Section~\ref{sec:prelims} we discussed \mindelay\ which finds the
minimum delay embedding of $\compgraph$ on $\net.$ The delay
calculation in \eqref{eq:nodedelay} (Section~\ref{sec:prelims}, Page
$3)$ assumes that each edge $l= (u_i,u_j) \in \netedges$ has infinite
capacity and can transmit as much data as needed in time $T(l),$ where
$T(l)$ is the delay of an edge $l.$ In general the edges in the
network have finite capacity and can transmit only one type of data in
time $T(l).$ Here we describe the delay in the network in the general
setting. Recall that an edge $\gamma \in \compedges$ is mapped to a
path in $\net$ in embedding $\embedding$ (defined in
Definition~\ref{def:embedding}). Delay of a network edge $l$ is the
time required for data to go from $u_i$ to $u_j.$ We say that an edge
$\gamma \in \compgraph$ has arrived at link $l$ when the data
corresponding to $\gamma$ is ready for transmission on $l$ at vertex
$u_i.$ Similarly, we say that the edge $\gamma$ has departed from link
$l$ when the data reaches $u_j$ via link $l.$ Let there be $k$ edges
$\gamma_1,\ldots,\gamma_k$ of $\compgraph$ mapped to an edge $l \in
\netedges$ under embedding $\embedding.$ Let the arrival time of these
edges at the link $l$ be $a_{\gamma_1}^l \leq a_{\gamma_2}^l \ldots
\leq a_{\gamma_k}^l.$ We assume that the capacity of each link is such
that at a given time only one kind of data can be transmitted over it
in the network. Then the departure time of these edges from the link
will be $z_{\gamma_1}^l < z_{\gamma_2}^l \ldots < z_{\gamma_k}^l.$ The
departure from the link $l$ can be calculated recursively as follows:
$z_{\gamma_i}^l := \max(z_{\gamma_{i-1}}^l,a_i^l) + T(l).$
Let an edge $\gamma =(\omega_i,\omega_j) \in \compedges$ be mapped to
a path $u_1,u_2,\ldots,u_m$ in $\net$ under embedding $\embedding$
such that $\embedding(\omega_i) = u_1$ and $\embedding(\omega_j) =
u_m.$ Then the delay of $\gamma$ in the embedding is the sum of the
delay incurred at each link $u_1u_2,u_2u_3,\ldots,u_{m-1}u_m$ which
can be written as:
\begin{equation}
  d(\embedding(\gamma)) :=
  d(\embedding(\omega_i),\embedding(\omega_j)) =
  z_{\gamma}^{u_{m-1}u_m} -
  a_{\gamma}^{u_1u_2}. \label{eq:actualdelay}
\end{equation} 
Now the delay of a vertex $\omega_j$ in the embedding $\embedding$ can
be defined as \eqref{eq:nodedelay} where
$d(\embedding(\omega_i),\embedding(\omega_j))$ in computed as the
above equation. We explain our point by an example.
\begin{example}
  \label{ex:delay}
  Consider a computation graph and a communication graph shown in
  Fig.~\ref{fig:delay}. We consider that there are processing delays
  and delay associated with each edge of the communication graph is
  $1.$ Consider an embedding $\embedding$ such that $\embedding(a) =
  k,\embedding(b) = l,\embedding(c) =m.$ Similarly, $\embedding(s_1a)
  = s_1-k,\embedding(s_2a) = s_2-i-j-k,\embedding(s_3b) =
  s_3-i-j-l,\embedding(s_4b) =s_4-l,\embedding(ac) =
  k-m,\embedding(bc)=l-m,\embedding(ct)=m-t.$ It is easy to observe
  that using the delay model described in the
  Section~\ref{sec:prelims} the delay of the embedding is $5.$ While
  in the model described above as the embedding of edges $s_2a$ and
  $s_3b$ have a common edge $i-j$ in them the delay of the embedding
  increases to $6.$
  \begin{figure}[tbp]
    \centering
     \subfloat[]{
\begin{tikzpicture}[>=latex,scale=0.7]
 \scriptsize
 \tikzstyle{every node} = [circle,draw=black]
 \node (a) at (-2,0) [fill=blue!20] {};
 \node (b) at (-1,0) [fill=blue!20] {};
 \node (h) at (1,0) [fill=blue!20] {};
 \node (c) at (2,0) [fill=blue!20] {};
 \node (d) at (-1,-1)[font=\tiny]  {$+$};
 \node (e) at (1,-1) [font=\tiny]{$+$};
 \node (f) at (0,-2)[font=\tiny] {$\times$};
 \node (g) at (0,-3) [fill=red!20]{};
 \node at (-2,0.4) [draw=none]{$s_1$};
 \node at (-1,0.4) [draw=none]{$s_2$};
 \node at (1,0.4) [draw=none]{$s_3$};
 \node at (2,0.4) [draw=none]{$s_4$};
 \node at (0,-3.4) [draw=none]{$t$};
 \node at (d) [draw=none,right=4pt] {$a$};
 \node at (e) [draw=none,right=4pt] {$b$};
 \node at (f) [draw=none,right=4pt] {$c$};
 \draw [->] (a) -- (d) ;
 \draw [->] (b) -- (d) ;
 \draw [->] (h) -- (e) ;
 \draw [->] (c) -- (e) ;
 \draw [->] (d) -- (f) ;
 \draw [->] (e) -- (f) ;
 \draw [->] (f) -- (g) ;
\end{tikzpicture}
}
\hspace*{10pt}
\subfloat[]{
\begin{tikzpicture}[>=latex,scale=0.7]
  \scriptsize
  \tikzstyle{every node} = [circle,draw=black]
  \node (a) at (0,2) [fill=blue!20] {};
  \node (b) at (0,1) {};
  \node (c) at (1,1) [fill=blue!20] {};
  \node (d) at (0,0) {};
  \node (e) at (-1,-1) {};
  \node (f) at (1,-1) {};
  \node (g) at (-1.5,0) [fill=blue!20] {};
  \node (h) at (1.5,0) [fill=blue!20] {};
  \node (i) at (0,-2) {};
  \node (j) at (0,-2.5) [fill=red!20] {};  
 \draw [-] (a) -- (b) -- (d) -- (e) -- (i) -- (j);
 \draw [-] (c) -- (b);
 \draw [-] (g) -- (e);
 \draw [-] (d) -- (f) -- (i);
 \draw [-] (h) -- (f); 
  \node at (a) [draw=none,above] {$s_2$};
  \node at (c) [draw=none,above] {$s_3$};
  \node at (g) [draw=none,above] {$s_1$};
  \node at (h) [draw=none,above] {$s_4$};
  \node at (j) [draw=none,below] {$t$};
  \node at (b) [draw=none,left] {$i$};
  \node at (d) [draw=none,left] {$j$};
  \node at (e) [draw=none,left] {$k$};
  \node at (f) [draw=none,left] {$l$};
  \node at (i) [draw=none,left] {$m$};
\end{tikzpicture}
} 
  \caption{(a) Computation graph for function $r(t) =
      (x_1+x_2)(x_3+x_4)$.  (b) Communication Network where every
      edge has delay $1$}
    \label{fig:delay}
  \end{figure}
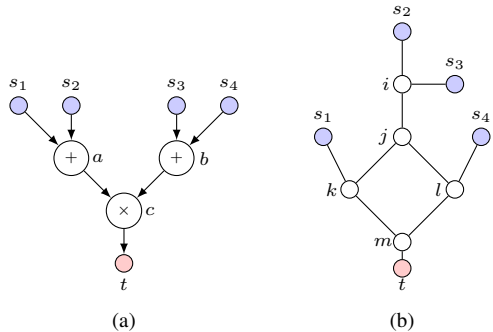
\end{example}

As mentioned in Example~\ref{ex:delay} the actual delay is more than
the delay defined by \eqref{eq:nodedelay} when multiple edges of
$\compgraph$ are mapped to one edge of $\net$ in an embedding. We
study the impact of our assumption via simulations.

We studied the behavior of minimum delay embedding and find the
statistics on maximum number of times an edge of $\net$ is used in the
minimum delay embedding of a typical $\compgraph.$ In our study the
computation graph was taken to be a binary tree of $p=32$ vertices and
its minimum delay embedding on a random graph of $n=120$ vertices was
calculated using Algorithm~\ref{algo:tree}. The probability of an edge
being present in the random graph ($p_r)$ was varied from $0.01$ to
$0.9.$ Note than as $p_r$ increases the number of edges in the network
increases. The communication and transmission costs were assumed to be
one and the processing cost was chosen uniformly from integers
$[1,10].$ For each value of $p_r,$ $32$ instances of network were
generated and for each instance Algorithm~\ref{algo:tree} was run for
$10$ random initial placements of sources and sink. The mean and the
median of the maximum number of times an edge in the network graph is
used in the embedding for each $p_r$ is shown in
Fig.~\ref{fig:meantree}. Observe that as the number of edges are
increased in the network maximum number of times an edge is used
converges to one. Hence we can say that for the networks with large
number of edges compared to that of the computation graph our
assumption of delay calculated by \eqref{eq:nodedelay} will be same as
that of delay computed by \eqref{eq:actualdelay}.

\begin{figure}[tbp]
  \centering
  \includegraphics[width=\linewidth]{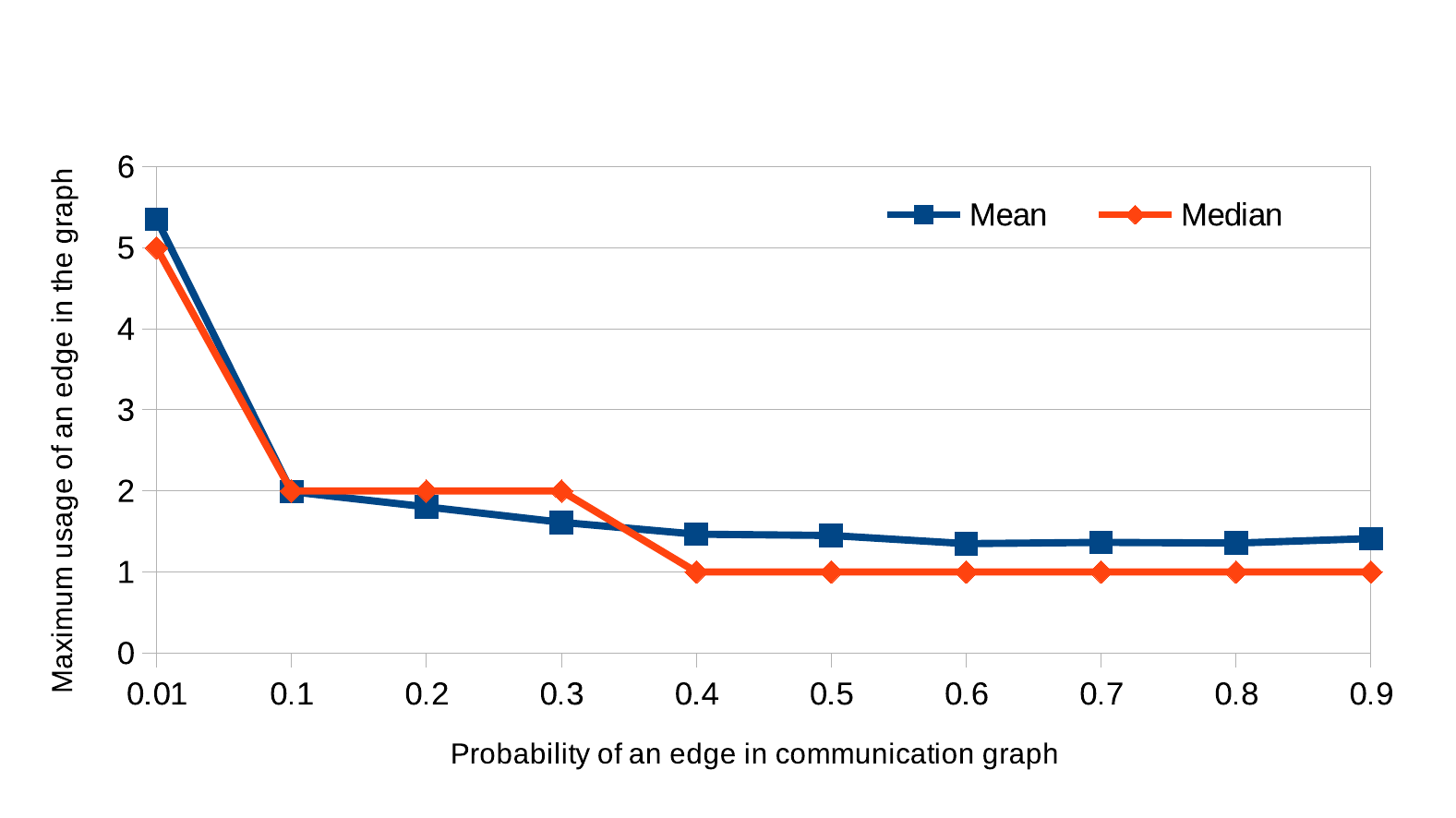}
  \caption{Statistics for maximum number of times a link is used in
    minimum delay embedding}
  \label{fig:meantree}
\end{figure}

\bibliographystyle{IEEE}
\bibliography{function-computation}

\appendices

\section{Hardness PCS-FM}
\label{sec:PCS-FM}

Here we prove that PCS-FM problem is NP-complete by reducing it to
another NP-complete problem \textit{Precedence Constraint Scheduling
  (PCS)} \cite{Garey79}.

PCS problem is defined as follows. Given a set $H^{\prime}$ of tasks
with a $\lessdot^{\prime}$ partial order on it each having length
$l(h) =1\ \ \forall h \in H^{\prime}$ and $m^{\prime} \in \mathbb{Z}^{+}$
processors then find a schedule $\sigma$ of tasks on processors which
meets an overall deadline $D^{\prime}$ and obeys the precedence
constraints, i.e., if for some $h_i,h_j \in H^{\prime}$ $h_i \lessdot
h_j$ then $\sigma(h_j) \geq \sigma(h_i) + 1.$

First we define an instance of PCS-FM problem $\phi =
(H,\lessdot,\tau,p(\tau),m,l,D)$ from an instance of PCS problem $\psi
= (H^{\prime},\lessdot^{\prime},m^{\prime},l^{\prime},D^{\prime}).$ We
create partial order graph $\lessdot$ from $\lessdot^{\prime}$ as
shown in Fig.~\ref{fig:scheduling}. The graph $\lessdot$ has two
parts: One part is the same as $\lessdot^{\prime}$ and the other part
has $k < m$ new vertices $s_1,\ldots,s_k$ giving $H = H^{\prime} \cup
\{s_1,\ldots,s_k\}.$ The vertices $s_1,\ldots,s_{k-1}$ are connected
to $s_k$ by a directed edge and $s_k$ is connected to all the vertices
of $\lessdot.$ Note that as all the edges are going away from $s_i$
the new graph is still a DAG. Let $m= m^{\prime},$ and $\tau =
\{s_1,\ldots,s_k\}.$ Define $l(s_i) = 1 \ \ \forall i \in [1,k]$ and
$l(h) = l^{\prime}(h)$ for all $h \in H^{\prime}.$ Let us number the
processors from $1$ to $m$ as $q_1, \ldots,q_m.$ Let us define $p(s_i)
:= q_i$ and the deadline for $\phi$ is $D = D^{\prime} + 2.$

To prove our claim we have to show that there exists a schedule
$\sigma^{\prime}$ for $\psi$ which meets the deadline $D^{\prime}$ if
and only if there exists a schedule $\sigma$ for $\phi$ which meets
the deadline $D.$ Now observe that in any schedule $\sigma$ for $\phi$
any task in $H^{\prime}$ cannot start unless task $s_k$ is finished
which in turn cannot start unless all the tasks $s_1,\ldots,s_{k-1}$
are finished. As it is given that the tasks $s_1,\ldots,s_{k-1}$ go to
separate processors (due to $p(s_i))$ they all can be finished in $1$
time step giving $\sigma(s_k) \geq \sigma(s_i)+1 = 1.$ Similarly
$\sigma(h) \geq \sigma(s_k) +1 = 2$ for all $h \in H^{\prime}.$ Hence
if there is a schedule $\sigma^{\prime}$ for $\psi$ which starts at
$0$ and finishes before $D$ then a schedule $\sigma$ for $\phi$ can be
defined as $\sigma(h) = \sigma^{\prime}(h) +2$ for all $h \in H,$
$\sigma(s_i) = 0$ for all $i \in [1,k-1]$ and $\sigma(s_k) = 1.$ It is
easy to observe that this is a valid possible schedule and finishes
before $D^{\prime}+2.$

\begin{figure}[tbp]
  \centering
\begin{tikzpicture}[>=latex]
\draw (0.5,1) rectangle (2,2);
\footnotesize;
\node at (1.2,1.7) {$\lessdot^{\prime}$};
\fill (1.5,0.5) circle (0.03cm);
\node at (1.9,0.5) {$s_k$};
\draw [->] (1.5,0.5) -- (0.9,1.2) [fill] circle (0.03cm);
\draw [->] (1.5,0.5) -- (1.6,1.8) [fill] circle (0.03cm);
\draw [->] (1.5,0.5) -- (1.9,1.2) [fill] circle (0.03cm);
\draw [->] (1.5,0.5) -- (1,0) [fill] circle (0.03cm);
\draw [->]  (1.5,0.5) -- (1.5,0) [fill] circle (0.03cm);
\draw [->] (1.5,0.5) -- (2,0) [fill] circle (0.03cm);
\node at (1,-0.3) {$s_1$};
\node at (1.5,-0.3) {$s_i$};
\node at (2,-0.3) {$s_{k-1}$};
\draw (0,-1) rectangle (2.5,2.5);
\node at (2.7,2.4) {$\lessdot$};
\end{tikzpicture}
 \caption{Transformation of PCS into PCS-FM problem}
  \label{fig:scheduling}
\end{figure}
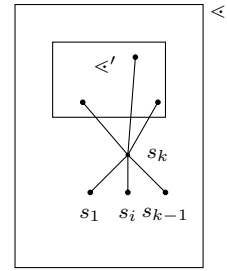

To complete the proof we have to prove that if there is no schedule
for $\psi$ which finishes before $D^{\prime}$ then there is no
schedule for $\phi$ which finishes before $D =D^{\prime}+2.$ We prove
this by contradiction. Let us assume that there is a schedule $\sigma$
for $\phi$ which finishes before $D$ but there is no schedule for
$\psi$ which finishes before $D^{\prime}.$ As noted earlier in any
schedule for $\phi$ first two time slots are required to finish tasks
$s_1,\ldots,s_k$ and then only any other task can start. So in the
schedule $ 2 \leq \sigma(h) \leq D$ for all $h \in H^{\prime}.$ Total
time taken to finish the tasks of $H^{\prime}$ is $ \leq D-2 =
D^{\prime} +2 -2 = D^{\prime}.$ It means that there is a mapping of
tasks of $H^{\prime}$ on $m$ processors such that the total time to
finish is less than $D^{\prime}$ which implies that there is a
schedule for $\psi$ which finishes before the deadline $D^{\prime}.$
This is a contradiction to our assumption. Hence it is proved that the
problem PCS-FM is as hard as PCS.

\end{document}